\documentclass[a4paper,USenglish,pdfa]{lipics-v2021}
\usepackage{mathtools}
\usepackage[table]{xcolor}
\usepackage{environ}
\usepackage{epsfig}
\usepackage{tikz}
\usetikzlibrary{backgrounds,matrix,shapes,arrows,arrows.meta,fit,positioning,decorations.pathmorphing,calc}

\nolinenumbers

\newtheorem{theoremSeparate}{Theorem}

\makeatletter
\newcommand{\problemtitle}[1]{\gdef\@problemtitle{#1}}
\newcommand{\probleminput}[1]{\gdef\@probleminput{#1}}
\newcommand{\problemquestion}[1]{\gdef\@problemquestion{#1}}
\NewEnviron{problem}{
  \problemtitle{}\probleminput{}\problemquestion{}
  \BODY
  \par\addvspace{.5\baselineskip}
  \noindent
  \begin{tabularx}{\textwidth}{@{\hspace{\parindent}} l X c}
      \multicolumn{2}{@{\hspace{\parindent}}l}{\textsc{\@problemtitle}} \\
    \textbf{Input:} & \@probleminput \\
    \textbf{Question:} & \@problemquestion
  \end{tabularx}
  \par\addvspace{.5\baselineskip}
}
\makeatother

\newcommand{\pTriple}{\textup{TRIPLE}}
\newcommand{\pUnif}{\textup{UNIF}}
\newcommand{\pCycs}{\textup{CYCS}}
\newcommand{\pTori}{\textup{TORI}}
\newcommand{\pOtori}{\textup{OTORI}}
\newcommand{\pCtori}{\textup{CTORI}}
\newcommand{\pTtori}{\textup{TTORI}}
\newcommand{\pFlip}{\textup{FLIP}}
\newcommand{\pSw}{\textup{SW}}
\newcommand{\pSat}{\textup{SAT}}
\newcommand{\pCol}{\textup{COL}}
\newcommand{\pCold}{\textup{COLD}}
\newcommand{\pProd}{\textup{PROD}}
\newcommand{\pPow}{\textup{POW}}
\newcommand{\pGesqrt}{\textup{GESQRT}}
\newcommand{\pLe}{\textup{LE}}
\newcommand{\pResc}{\textup{RESC}}
\newcommand{\pRes}{\textup{RES}}
\newcommand{\pResThree}{\textup{RES3}}
\newcommand{\pEqres}{\textup{EQRES}}

\newcommand{\mysymbol}[1]{{\mbox{\raisebox{-0.25em}{\epsfysize=1.1em\epsfbox{#1}}}}}

\newcommand{\tileC}{\mysymbol{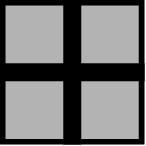}}
\newcommand{\tileH}{\mysymbol{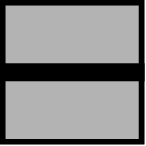}}
\newcommand{\tileV}{\mysymbol{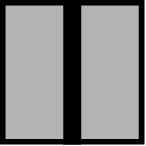}}

\newcommand{\tileZeroA}{\mysymbol{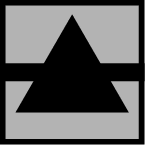}}
\newcommand{\tileOneA}{\mysymbol{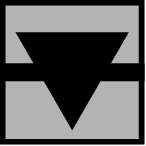}}

\newcommand{\wangTile}[6]{
        \node[minimum size=#5] (#6) {};
        \draw [black, fill=#1] (#6.north west)--(#6.north east)--(#6.center)--cycle;
        \draw [black, fill=#2] (#6.north east)--(#6.south east)--(#6.center)--cycle;
        \draw [black, fill=#3] (#6.south east)--(#6.south west)--(#6.center)--cycle;
        \draw [black, fill=#4] (#6.south west)--(#6.north west)--(#6.center)--cycle;
        \xglobal\colorlet{#6.topcolor}{#1}
        \xglobal\colorlet{#6.rightcolor}{#2}
        \xglobal\colorlet{#6.bottomcolor}{#3}
        \xglobal\colorlet{#6.leftcolor}{#4}
}

\newcommand{\hiddenWangTile}[6]{
        \node[draw=none, minimum size=#5] (#6) {};
        \xglobal\colorlet{#6.topcolor}{#1}
        \xglobal\colorlet{#6.rightcolor}{#2}
        \xglobal\colorlet{#6.bottomcolor}{#3}
        \xglobal\colorlet{#6.leftcolor}{#4}
}

\newcommand{\wangTileClone}[3]{
        \node[minimum size=#2] (#3) {};
        \draw [black, fill=#1.topcolor] (#3.north west)--(#3.north east)--(#3.center)--cycle;
        \draw [black, fill=#1.rightcolor] (#3.north east)--(#3.south east)--(#3.center)--cycle;
        \draw [black, fill=#1.bottomcolor] (#3.south east)--(#3.south west)--(#3.center)--cycle;
        \draw [black, fill=#1.leftcolor] (#3.south west)--(#3.north west)--(#3.center)--cycle;
}

\newcommand{\noWangTile}[2]{
        \node[minimum size=#1, draw=none] (#2) {};
        \xglobal\colorlet{#2.topcolor}{white}
        \xglobal\colorlet{#2.rightcolor}{white}
        \xglobal\colorlet{#2.bottomcolor}{white}
        \xglobal\colorlet{#2.leftcolor}{white}
}

\newcommand{\tileXA}[2]{\wangTile{cyan}{red}{cyan}{cyan}{#1}{#2}}
\newcommand{\tileXB}[2]{\wangTile{red}{green}{green}{red}{#1}{#2}}
\newcommand{\tileXC}[2]{\wangTile{cyan}{cyan}{cyan}{green}{#1}{#2}}
\newcommand{\tileXD}[2]{\wangTile{green}{cyan}{red}{red}{#1}{#2}}
\newcommand{\tileXE}[2]{\wangTile{red}{cyan}{green}{cyan}{#1}{#2}}

\newcommand{\hiddenTileXA}[2]{\hiddenWangTile{cyan}{red}{cyan}{cyan}{#1}{#2}}
\newcommand{\hiddenTileXB}[2]{\hiddenWangTile{red}{green}{green}{red}{#1}{#2}}
\newcommand{\hiddenTileXC}[2]{\hiddenWangTile{cyan}{cyan}{cyan}{green}{#1}{#2}}
\newcommand{\hiddenTileXD}[2]{\hiddenWangTile{green}{cyan}{red}{red}{#1}{#2}}
\newcommand{\hiddenTileXE}[2]{\hiddenWangTile{red}{cyan}{green}{cyan}{#1}{#2}}

\hideLIPIcs
\title{On the Complexity of the Conditional Independence Implication Problem With Bounded Cardinalities}
\titlerunning{On the Complexity of the CI Implication Problem With Bounded Cardinalities}
\author{Michał Makowski}
{Faculty of Mathematics, Informatics and Mechanics, University of Warsaw, Warsaw, Poland}
{mm406231@students.mimuw.edu.pl}
{}
{}
\authorrunning{M. Makowski}
\Copyright{Michał Makowski}
\acknowledgements{The author would like to thank Damian Niwiński for his guidance and the anonymous reviewers for their valuable comments.}
\relatedversiondetails{Full version of a paper presented at MFCS 2024.}{}

\begin{document}
\maketitle

\begin{abstract}
We show that the conditional independence (CI) implication problem with bounded cardinalities, which asks whether a given CI implication holds for all discrete random variables with given cardinalities, is co-NEXPTIME-hard. The problem remains co-NEXPTIME-hard if all variables are binary.
The reduction goes from a variant of the tiling problem and is based on a prior construction used by Cheuk Ting Li to show the undecidability of a related problem where the cardinality of some variables remains unbounded. The CI implication problem with bounded cardinalities is known to be in EXPSPACE, as its negation can be stated as an existential first-order logic formula over the reals of size exponential with regard to the size of the input. 
\end{abstract}

\section{Introduction}
The implication problem for conditional independence statements is one of the major  decision problems  arising in multivariate statistical modeling and other applications~\cite{geiger1993logical}.   The problem asks whether a list  of conditional independence statements implies another such statement (see the exact formulation below).  The problem is a special case of 
the conditional entropic inequality problem, as the statement 
\emph{$X$ and $Y$  are independent, given $Z$} (sometimes denoted
$X \perp Y \mid Z $)
is equivalent to the equation $I (X; Y | Z) = 0$ in information theory.
Here the random variables in consideration are of the form $(X_{i_1}, \ldots , X_{i_{\ell }})$, abbreviated by $ X_Z$  with
$ Z = \{ i_1, \ldots , i_{\ell }\} $, selected from a fixed $ n$-tuple 
of variables $X_1, \ldots, X_n$ considered with a joint distribution. As with the general problem, this can be considered 
over continuous,
infinite discrete or finite discrete random variables.
Furthermore, the CI implication problem can be refined by imposing certain requirements on the sets $A, B, C$ in $X_A \perp X_B \mid X_C$, e.g., 
they must be pairwise disjoint for \emph{disjoint CI}, and for \emph{saturated CI} they must additionally satisfy $A \cup B \cup C = \{1, \ldots, n\}$.
We will focus on discrete random variables with a finite domain and without constraints on the sets, addressing disjoint
CI in Section~\ref{section:disjoint}.

If the domain size is bounded, this problem is decidable since the conditional independence can be expressed as an arithmetic formula in terms
of elementary events' probabilities. Considering all possible domain sizes yields a semi-algorithm for
finding a counter-example to the implication, showing that the unbounded problem is co-recursively enumerable (as noted by Khamis \emph{et al}~\cite{khamis2020decision}). The decidability
of the general CI implication problem was unknown for a long time, with only special cases resolved. Finally, Cheuk Ting Li published
two papers \cite{li2021undecidabilityitw,li2022undecidability} proving the problem to be undecidable. This was also shown independently by
K{\"u}hne and Yashfe~\cite{kuhne2022entropic}.  Still unknown is the complexity of the bounded problem --
the method of constructing an existential formula of Tarski's arithmetic yields an upper bound of EXPSPACE (cf.~\cite{canny1988pspace}),  while the  other published algorithms appear to be mostly
heuristics. Hannula \emph{et al}~\cite{hannula2019facets} conjectured  that the problem can be actually easier, especially in the case where all variables are binary, as the arithmetic formula in question is of a very special form.  

In this paper we show that the problem is in general co-NEXPTIME-hard, and the hardness result continues to hold if all variables are binary.  Our  reduction is an adaptation of the
construction presented by Li~\cite{li2021undecidabilityitw} to show that the problem is undecidable if the cardinalities of some variables are bounded.
While there is still a gap between the lower and upper bound, this shows that the complexity of the CI implication problem is harder than it might have been expected.
\section{Problem statement}

Denote by $\text{card}(X)$ the cardinality of a random variable $X$. Formally, the following problem will be considered:
\begin{problem}
  \problemtitle{Bounded CI Implication}
  \probleminput{Integers $m, n$, given in unary. A list of $m+1$ triples $(A_i, B_i, C_i)$ of subsets of $\{1, \ldots, n\}$. A list of $n$ integers $K_j$, given in binary.
  }
  \problemquestion{Determine whether the implication
      \[
          \bigwedge_{i \in \{1, \ldots, m\}} (I(X_{A_{i}}; X_{B_{i}} | X_{C_{i}}) = 0) \Rightarrow I(X_{A_{m+1}}; X_{B_{m+1}} | X_{C_{m+1}}) = 0
      \]
      holds for all jointly distributed random variables $(X_1, \ldots, X_n)$ with $\text{card}(X_j) \leq K_j$ for all $j \in \{1, \ldots, n\}.$
  }
\end{problem}

We define \textsc{Constant-bounded CI Implication} as a variant of the above problem in which all $K_i$ are fixed to be equal
to 2 rather than given as input. We show the following:
\begin{theoremSeparate} \label{thm:hard}
    \textsc{Bounded CI Implication} and \textsc{Constant-bounded CI Implication} are co-NEXPTIME-hard. This also holds in the
    disjoint CI case, i. e. when for each $i$ the sets $A_i, B_i, C_i$ are pairwise disjoint.
\end{theoremSeparate}
We focus on the first part of the theorem -- the proof of the second part differs little from that given by Li~\cite{li2021undecidabilityitw} and is given in Section~\ref{section:disjoint}.

In order to state the tiling-based problems utilized in the reduction, we introduce some definitions based on those in~\cite{lewis1998elements}.
We define a \emph{tiling system} as a triple ${\mathcal{D} = (D, H, V)}$, where $D$ is a finite set
of tiles and $H, V \subseteq D^2$ are the horizontal and vertical constraints, accordingly, which give the pairs of tiles
that may be neighbors.
This is a generalization of Wang tiles, where a set of colors $C$ is given and tiles (formally quadruples from the set $C^4$) are represented by squares with colored edges with the requirement that only edges of the same color may touch.
As stated in~\cite{boas1996tilings}, Wang tiles correspond exactly to those tiling systems for which the implication
${(a \, R \, b \land a \, R \,  c \land d \, R \,  b) \Rightarrow d \, R \, c}$ holds for all $a, b, c, d \in D$ and both $R \in \{H, V\}$.

We define a $k \times l$ \emph{tiling} by $\cal{D}$ as a function $f : \{0, \ldots, k-1\} \times \{0, \ldots, l-1\} \to D$ such that:
\begin{itemize}
    \item $(f(m, n), f(m+1, n)) \in H$ for all $m < k - 1, n < l$,
    \item $(f(m, n), f(m, n+1)) \in V$ for all $m < k, n < l - 1$.
\end{itemize}
A \emph{periodic tiling} is one that also has ${(f(k-1, n), f(0, n)) \in H}$ and ${(f(m, l-1), f(m, 0)) \in V}$ for all $m < k, n < l$.
For a (non-periodic) $k \times l$ tiling $f$, the \emph{starting tile} and \emph{final tile} are the values of $f(0, 0)$ and $f(k-1, l-1)$, respectively.

We will show a polynomial-time many-one reduction from the following problem, which is known to be NEXPTIME-complete \cite[exercise 7.2.2.]{lewis1998elements}:
\begin{problem}
  \problemtitle{Binary Bounded Tiling}
  \probleminput{A tiling system $\cal{D}$, a starting tile $d_0 \in D$, an integer $k$ given in binary.}
  \problemquestion{Determine whether there exists a $k \times k$ tiling $f$ by $\cal{D}$ such that $f(0, 0) = d_0$.}
\end{problem}
The reduction consists of two parts, the first being a purely tiling-based reduction from the above to the following intermediate problem:
\begin{problem}
  \problemtitle{Periodic Bounded Tiling}
  \probleminput{A tiling system $\mathcal{D}$, a designated tile $t$, integers $m, n$ given in binary.}
  \problemquestion{Determine whether there exists a periodic tiling by $\mathcal{D}$ of size at most $m \times n$ which
                uses tile $t$.}
\end{problem}

The second part is a reduction from \textsc{Periodic Bounded Tiling} to the complement of \textsc{Bounded CI Implication}, based on a construction by Li~\cite{li2021undecidabilityitw}.

\section{First part of the reduction}
We first show a polynomial-time many-one reduction from \textsc{Binary Bounded Tiling} to \textsc{Periodic Bounded Tiling}.
This means that given a tiling system $\mathcal{D}$, starting tile $d_0$ and integer $k$, we will construct in polynomial
time a tiling system $\mathcal{D}''$, designated tile $t$ and integers $m, n$
such that \textsc{Binary Bounded Tiling} gives a positive answer for input $(\mathcal{D}, d_0, k)$ iff \textsc{Periodic Bounded Tiling} 
gives a positive answer for input $(\mathcal{D}'', t, m, n)$. This consists of two steps:
\begin{enumerate}
    \item Modify $\mathcal{D}$ into system $\mathcal{D}'$ such that valid tilings by $\mathcal{D}$ of size $k \times k$
        correspond to valid tilings by $\mathcal{D}'$ with certain corner constraints.
    \item Modify $\mathcal{D}'$ into system $\mathcal{D}''$ and tile $t$ such that valid tilings by $\mathcal{D}'$ with the above corner constraints correspond
        to periodic tilings by $\mathcal{D}''$ of size $(k+1) \times (k+1)$ that use tile $t$. \label{corner-constraints}
\end{enumerate}
This is a fairly typical reduction between tilings, similar to problems considered for instance in~\cite{gottesman2009quantum}.

\subsection*{Limiting tiling size}
For the first step, we create a tiling system $\mathcal{C}$ (of polynomial size with regard to the length of $k$), along with starting tile $c_0$ and final tile $c_1$, implementing
a binary counter that counts down from an appropriately chosen $k' \leq k$ (close to $k$) and whose position shifts by 1 with each decrement. The tile $c_1$ occurs when
the counter reaches 0. Similar constructions have been shown~\cite{rothemund2000squares}, our example is given in Figure \ref{fig:counter}. Thus, any tiling by $\mathcal{C}$ with $c_0$ in the top-right corner and $c_1$ in the bottom-left must be of size exactly $k \times k$.

\begin{figure}[p]
    \begin{centering}
        \begin{tabular}{|c|c|c|c|c|c|c|c|c|c|}
        \hline
            &   &   &   &   &   & \cellcolor{lightgray} ${\ }_1$ & \cellcolor{lightgray} $1_0$ & \cellcolor{lightgray} $0_1$ & \cellcolor{lightgray} $1_{\ }$ \\ \hline
            &   &   &   &   & ${\ }_1$ & $1_0$ & $0_0$ & $0_{\ }$ &   \\ \hline
            &   &   &   & ${\ }_0$ & $0_1$ & $1^*_1$ & $1^*_{\ }$ &   &   \\ \hline
            &    &   & ${\ }_0$ & $0_1$ & $1_0$ & $0_{\ }$ &   &   &   \\ \hline
            &   & ${\ }_0$ & $0_0$ & $0_1$ & $1^*_{\ }$ &   &   &   &   \\ \hline
            & ${\ }_0$ & $0_0$ & $0_0$ & $0_{\ }$ &   &   &   &   &   \\ \hline
            $\star$ & $1^*_1$ & $1^*_1$ & $1^*_{\ }$ &   &   &   &   &   &   \\ \hline
        \end{tabular}
        \caption{
            A tiling implementation of a binary counter which shifts its position on every decrement, with the starting
            and final tile in the top-right and bottom-left corners respectively. This example counts down from 5 (101 in binary).
            Every tile except for the final $\star$ is of the form $a_b^c$, where $a$ is the bit value (possibly blank), $b$ is the value
            of the bit directly to the right (or blank if there is none), and $c$ is optionally $*$ if a borrow operation is required.
            The shaded tiles of the top row function in the same manner, but they are ``memorized'' within the tiling system
            such that the placement of the top-right tile forces the top row to write out the binary initial value.
            The tiles in the lower rows are chosen
            deterministically based on their right and top neighbor. Finally, the $\star$ tile only occurs when the tile  above is blank and
            the one to the right requires a borrow, which indicates that the counter has just gone below zero.
            In order to be unable to further count down, we disallow any tiles being below or to the left of
            tile $\star$. The size of the tiling is $(k' + b + 2) \times (k' + 2)$, where $b$ is the number of bits and $k'$
            is the initial value; however, this could be modified such that the final tiling has size $(k' + b + 2) \times (k' + b + 2)$
            by padding with $b$ dummy rows at the top. For sufficiently large $k$, we can always efficiently find $b, k'$ such that $k' + b + 2 = k$.
        }
        \label{fig:counter}
    \end{centering}
\end{figure}

Consider a ``layering'' of $\mathcal{D}$ and $\mathcal{C}$ into system
$\mathcal{D} \times \mathcal{C} = (D \times C, H_{D \times C}, V_{D \times C})$, where
the relations are defined as $R_{D \times C} = \{ ((d, c), (d', c')) : (d, d') \in R_D \land (c, c') \in R_C\}$ for both $R \in \{H, V\}$.
This way, any tiling by $\mathcal{D} \times \mathcal{C}$ corresponds to a pair of tilings by $\mathcal{D}$
and $\mathcal{C}$.
We let $\mathcal{D}' = \mathcal{D} \times \mathcal{C}$ and define the corner constraints mentioned in point \ref{corner-constraints} by restricting the possible starting and final tiles to
${S = \{(d_0, c_0)\}}, {F = \{(d, c_1) : d \in D\}}$ respectively. This yields tilings by $\mathcal{D}'$ which
consist of a tiling by $\mathcal{D}$ with starting tile $d_0$ and a tiling by $\mathcal{C}$ with starting tile $c_0$ and final tile $c_1$.

\subsection*{Periodic tilings}
The second step realizes the above constraint while also converting between periodic and non-periodic tilings. It consists of adding 5 new tiles to $\mathcal{D}'$ ---
$\tileC$, $\tileH$, $\tileV$, $\tileZeroA$, $\tileOneA$ --- yielding $\mathcal{D}''$. The added constraints are shown in Figure \ref{fig:border}
and an example tiling in Figure \ref{fig:example}. The distinguished tile $t$ is $\tileC$ -- the idea is that its usage forces a ``border''
of $\tileH$ and $\tileV$ tiles. Within each of these borders (there can be multiple if $\tileC$ is used more than once)
there is a valid tiling by the system $\mathcal{D}'$ additionally satisfying the starting and final constraints $S, F$.

\begin{figure}[p]
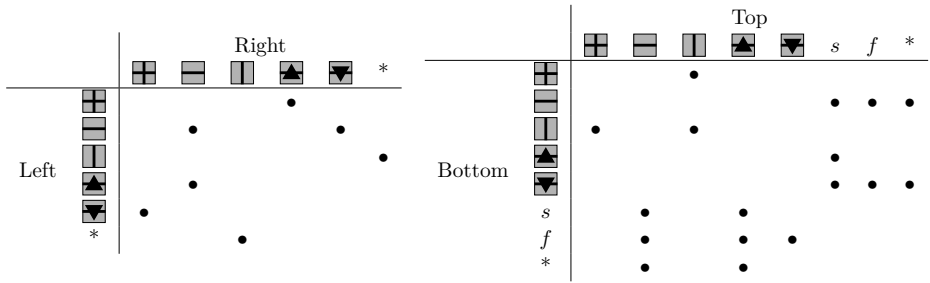

    \begin{centering}
        \scalebox{0.8}{
        \begin{tabular}{cc|cccccc}
            & & \multicolumn{6}{c}{Right} \\
                                    & & $\tileC$ & $\tileH$ & $\tileV$ & $\tileZeroA$ & $\tileOneA$ & * \\ \hline
            \multirow{6}{*}{Left} & $\tileC$ &   &   &   & $\bullet$ &   & \\
                                  & $\tileH$ &   & $\bullet$ &   & & $\bullet$ & \\
                                  & $\tileV$ &   &   &   &   &   & $\bullet$ \\
                                  & $\tileZeroA$ &   & $\bullet$ &   &   &   & \\
                                  & $\tileOneA$ & $\bullet$ &   &   &   &   & \\
                                  & * &   &   & $\bullet$ &   &   & \\
        \end{tabular}
        }
        \scalebox{0.8}{
        \begin{tabular}{cc|cccccccc}
            & & \multicolumn{8}{c}{Top} \\
                                  & & $\tileC$ & $\tileH$ & $\tileV$ & $\tileZeroA$ & $\tileOneA$ & $s$ & $f$ & * \\ \hline
            \multirow{8}{*}{Bottom} & $\tileC$ &   &   & $\bullet$ &   &   &   &   &   \\
                                    & $\tileH$ &   &   &   &   &   & $\bullet$ & $\bullet$ & $\bullet$ \\
                                    & $\tileV$ & $\bullet$ &   & $\bullet$ &   &   &   &   &   \\
                                    & $\tileZeroA$ &   &   &   &   &   & $\bullet$ &   &   \\
                                    & $\tileOneA$ &   &   &   &   &   &   $\bullet$ & $\bullet$ &   $\bullet$ \\
                                    & $s$ &   & $\bullet$ &   & $\bullet$ & &   &   &   \\
                                    & $f$ &   & $\bullet$ &   & $\bullet$ & $\bullet$ &   &   &   \\
                                    & * &   & $\bullet$ &   & $\bullet$ & &   &   &   \\
        \end{tabular}
        }
        \caption{
            Modified adjacency relation for the system $\mathcal{D}''$ -- only adjacencies marked by $\bullet$ are permitted, as well as all adjacencies from the original tiling system.
            The asterisk denotes any tiles from the system $\mathcal{D}'$, while $s, f$ represent any tile
            from the initial and final subset of tiles, respectively ($S$ and $F$ defined above).
        }
        \label{fig:border}
    \end{centering}
\end{figure}

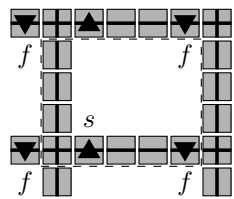
\begin{figure}[p]
    \centering
    \begin{tikzpicture}
      \matrix (M) [
        matrix of nodes, column sep=-0.2cm, row sep=-0.2cm
      ]
      {
        \tileOneA & \tileC & \tileZeroA & \tileH & \tileH & \tileOneA & \tileC \\
        $f$ &  \tileV& & & & $f$ & \tileV \\
                 & \tileV & & & & & \tileV \\
                 & \tileV & $s$ & & & & \tileV \\
        \tileOneA & \tileC & \tileZeroA & \tileH & \tileH & \tileOneA & \tileC \\
        $f$ & \tileV & & & & $f$ & \tileV \\
      };
    \node[draw=black,fit=(M-2-2)(M-5-6),dashed,inner sep = -3pt] {};
    \end{tikzpicture}
     \caption{
         An example ``border'' created by the above tiling, with tiles from the original set not shown and $s, f$ representing tiles from $S, F$ respectively. The dashed rectangle
         represents the actual rectangle being tiled, while the tiles outside are periodic copies added to better illustrate
         the construction.
     }
     \label{fig:example}
\end{figure}

\subsection*{Final conversion}\label{tiles-final}
Combining the above steps, a tiling by $\mathcal{D}$ of
size $k \times k$ is represented by a periodic $(k + 1) \times (k + 1)$ tiling by $\mathcal{D}''$ and thus we let $m = n = k + 1$.
The designated tile is set to $\tileC$, completing the reduction from \textsc{Binary Bounded Tiling}
to \textsc{Periodic Bounded Tiling}.

\subsection*{Variant with only powers of two} \label{section:powers-of-two}
Consider the following variant of the tiling problem:
\begin{problem}
  \problemtitle{Power-of-two Periodic Bounded Tiling}
  \probleminput{Integers $m, n$ given in unary, a tiling system $\mathcal{D}$, a designated tile $t$.}
  \problemquestion{Determine whether there exists a periodic tiling by $\mathcal{D}$ of size at most $2^m \times 2^n$ which
                uses tile $t$.}
\end{problem}
Note that this is the same as taking the input in binary while restricting it only to powers of two. This variant is also NEXPTIME-hard
because in the above reduction, the only possible sizes of tiling by the constructed tiling system are multiples of $k+1$,
both in width and height. Letting $m = n = \lceil \log_2(k + 1) \rceil$, we have $k + 1 \leq 2^m < 2(k + 1)$ and the same
for $2^n$. Therefore, the possible tilings are the same for size bound $(k + 1) \times (k + 1)$ and $2^{\lceil \log_2(k + 1) \rceil} \times 2^{\lceil \log_2(k + 1) \rceil}$.

\section{Second part of the reduction}\label{second-part}
Given a tiling system $\mathcal{D}$, a designated tile $t$ and integers $m, n$ given in binary, we will construct
(in polynomial time) a CI implication with bounded cardinalities which does not hold iff \textsc{Periodic Bounded Tiling} has a solution for input $(\mathcal{D}, t, m, n)$.
This is based on the construction of Li~\cite{li2021undecidabilityitw}. The main changes to Li's construction are as follows:
\begin{itemize}
    \item provide bounds for the random variables used in the construction -- this is done as the implication is being constructed;
    \item reduce the size of the implication from exponential to polynomial with regard to the input -- only one part (the predicate $\pCol$)
        needs to be replaced by a polynomial-size equivalent;
    \item modify the representation of tiles to better suit bounded size and non-Wang tiles;
    \item add the requirement of the usage of a given tile -- this is done by modifying the consequent of the implication.
\end{itemize}

\subsection{Preliminaries}
We denote by $\text{Unif}(S)$ a uniform distribution over set $S$ and by $\text{Bern}(p)$ a Bernoulli distribution with parameter $p$.
We use the shorthand $X^k$ to represent a tuple of random variables $(X_1, \ldots, X_k)$, which is in itself also a random variable.

The \emph{entropy} of a finite discrete random variable $X$ with domain $\mathcal{X}$ is defined as
\[
    H(X) = -\sum_{x \in \mathcal{X}} \mathbf{P}(X = x) \log \mathbf{P}(X = x).
\]
The \emph{entropy of $X$ conditioned on $Y$}, with $X, Y$ finite discrete random variables with domains $\mathcal{X}, \mathcal{Y}$ respectively, is defined as
\[
    H(X|Y) = -\sum_{x \in \mathcal{X}, y \in \mathcal{Y}} \mathbf{P}(X = x \land Y = y) \log \frac{\mathbf{P}(X = x \land Y = y)}{\mathbf{P}(X = x)}.
\]
Finally, the \emph{conditional mutual information of $X$ and $Y$ given $Z$} can be defined as
\[
    I(X; Y | Z) = H(X|Z) - H(X|Y, Z).
\]
Similarly, the \emph{mutual information} of $X$ and $Y$ is defined as
\[
    I(X; Y) = H(X) - H(X|Y).
\]

Note that the conditional independence (CI) statement $I(X; Y | Z) = 0$ is equivalent to the fact that $X$ and $Y$ are independent given $Z$,
and thus can be expressed without the usage of logarithms as
\[
    \mathbf{P}(X = x \land Y = y \land Z = z)\mathbf{P}(Z = z) = \mathbf{P}(X = x \land Z = z)\mathbf{P}(Y = y \land Z = z)
\]
for all $x \in \mathcal{X}, y \in \mathcal{Y}, z \in \mathcal{Z}$. Further, the functional dependence statement $H(X | Y) = 0$,
which states that for any $y \in \mathcal{Y}$, there exists exactly one $x \in \mathcal{X}$ such that $\mathbf{P}(X = x \land Y = y) \neq 0$
can be expressed equivalently as a (non-disjoint) CI statement $I(X; X | Y) = 0$.

We will follow Li's construction~\cite{li2021undecidabilityitw}, providing cardinality bounds and modifications where necessary.
While the final goal is a CI implication, we will mostly construct \emph{affine existential information predicates}
(AEIP)~\cite{li2021undecidabilityitw}, converting to a CI implication at the end. We will only
consider a special form of AEIP which consists of an existentially quantified conjunction of CI statements.
This family of predicates is closed under conjunction, in particular we can use a predicate
within the definition of another predicate (implicitly renaming variables in the case of a naming conflict).

Since our goal is to construct a bounded CI implication, every quantified
variable will be given a \emph{cardinality bound} -- the maximum allowed size of its domain. We denote the existential
quantification of a variable $X$ with $\text{card}(X) \leq k$ by the shorter notation $\exists X \leq k$, similarly for tuple variables
$\exists X^2 \leq k$ represents the existence of variables $X_1, X_2$ with $\text{card}(X_1), \text{card}(X_2) \leq k$. Whenever a predicate
takes arguments, their cardinalities are already bounded since they have been quantified. We may need to refer to
these bounds when quantifying new variables, denoting by $K_X$ the bound already given to variable $X$.
Finally, whenever $\leq$ is replaced by $\leq_i$, this indicates an ``implicit'' bound, that is one which does not change
the meaning of the predicate because it is already satisfied by any such quantified variable even without the explicit bound.
An example of this is that whenever $H(X | Y) = 0$, $X$ is functionally dependent on $Y$ and so $X \leq_i K_Y$.

The first defined predicate is $\pTriple$:
\begin{align*}
 \pTriple(Y_1, Y_2, Y_3): \ & H(Y_1 | Y_2, Y_3) = H(Y_2 | Y_1, Y_3) = H(Y_3 | Y_1, Y_2) = 0\\
                                 & \land I(Y_1; Y_2) = I(Y_1; Y_3) = I(Y_2; Y_3) = 0.
\end{align*}
By definition, predicate $\pTriple$ of three variables $Y_1, Y_2, Y_3$ is satisfied iff $Y_1, Y_2, Y_3$ are pairwise independent
and each variable is functionally dependent on the other two. Functional dependency is a special case of conditional
independence, since $H(X|Y) = I(X; X | Y)$. The following is shown in~\cite{zhang1997ineq}:
\begin{lemma}
    \label{lemma:triple}
    If $\pTriple(X, Y, Z)$ is satisfied, then $X, Y, Z$ are all uniformly distributed and have the same cardinality.
\end{lemma}
\begin{proof}
    Consider any $x, y$ such that $\mathbf{P}(X = x) > 0, \mathbf{P}(Y = y) > 0$.
    Since $Z$ is functionally dependent on $X, Y$, there exists a unique $z$ such that $\mathbf{P}(X = x \land Y = y \land Z = z) > 0$.
    Since we also have $X$ functionally dependent on $Y, Z$ and $Y$ on $X, Z$, we can write
    \[
         \mathbf{P}(X = x \land Z = z) = \mathbf{P}(X = x \land Y = y \land Z = z) = \mathbf{P}(Y = y \land Z = z).
    \]
    Since $Y$ and $Z$ are independent and so are $X$ and $Z$, we have
    \[
        \mathbf{P}(X = x)\mathbf{P}(Z = z) = \mathbf{P}(X = x \land Y = y \land Z = z) = \mathbf{P}(Y = y)\mathbf{P}(Z = z).
    \]
    We can divide both sides by $\mathbf{P}(Z = z) > 0$, yielding $\mathbf{P}(X = x) = \mathbf{P}(Y = y)$. Repeating this
    symmetrically, we get that for any $x, y, z$ such that $\mathbf{P}(X = x) > 0$, $\mathbf{P}(Y = y) > 0$, $\mathbf{P}(Z = z) > 0$,
    we have $\mathbf{P}(X = x) = \mathbf{P}(Y = y) = \mathbf{P}(Z = z)$. Therefore $X, Y, Z$ are all uniform and their supports
    have equal sizes.
\end{proof}
This is used in the next predicate, $\pUnif$:
\[
    \pUnif(X) : \exists U_1 \leq_i K_X, U_2 \leq_i K_X : \pTriple(X,U_1,U_2).
\]
By definition, predicate $\pUnif$ of one variable $X$ is satisfied iff there exist discrete random variables $U_1, U_2$ jointly distributed
with $X$ such that $\pTriple(X, U_1, U_2)$ holds,
which in turn is equivalent to $X$ being uniformly distributed over its support. Lemma~\ref{lemma:triple} immediately shows
that the implicit cardinality bounds for $U_1, U_2$ are correct.

For any constant $k$, predicate $\pUnif_k(X)$ is defined to imply $X$ being uniformly distributed over a domain of size $k$:
\begin{align*}
 & \pUnif_k(X) : \pUnif(X) \land \alpha_{k} \leq H(X) \leq \alpha_{k+1},
\end{align*}
where $\alpha_k \in \mathbb{Q}$ is some rational with $\log(k-1) < \alpha_k < \log k$ (because the entropy of a uniform
variable is the logarithm of the cardinality of its support). While this is still a valid
AEIP, it is not a conjunction of CI statements. However, in the bounded cardinality setting $\pUnif_k$ can be restated in this form, as shown in Section~\ref{section:unif} as well as ~\cite{li2021undecidabilityitw}.  Note that this predicate imposes an exact domain size constraint, while the cardinality bounds given as input provide only an upper bound.
Finally, note that the predicates $\pUnif$, $\pUnif_k$ are satisfiable -- for any $k > 0$, there exists a variable $X$
which satisfies $\pUnif_k(X)$ and so $\pUnif(X)$.

Define the \emph{characteristic bipartite graph} of random variables $X_1, X_2$ with (disjoint) supports $\mathcal{X}_1, \mathcal{X}_2$
as the undirected graph with set of vertices $V = \mathcal{X}_1 \cup \mathcal{X}_2$ and set of edges $E = {\{(x_1, x_2) : {x_1 \in \mathcal{X}_1}, {x_2 \in \mathcal{X}_2}, \mathbf{P}(X_1 = x_1 \land X_2 = x_2) > 0\}}$.

Li constructs the predicate
\begin{align*}
\pCycs(X_1, X_2) : \ & \exists U \leq_i 2: \pUnif(X_1) \land \pUnif(X_2) \land \pUnif_2(U) \\
 & \land I(X_1; U) = I(X_2; U) = 0 \\
 & \land H(X_1 | X_2, U) = H(X_2 | X_1, U) = 0 \\
 & \land H(U | X_1, X_2) = 0
\end{align*}

and shows the following (without cardinality bounds):
\begin{lemma} $\pCycs(X_1, X_2)$ is satisfied iff $X_1, X_2$ are uniform and the characteristic bipartite graph of
    $X_1, X_2$ consists only of vertex-disjoint simple cycles.
\end{lemma}
\begin{proof}
    Clearly, a graph consists only of disjoint simple cycles iff the degree of each vertex is either zero or two.
    For the ``only if'' direction, fix $X_1 = x_1$. Then $X_2$ is functionally dependent on $U$ (since $H(X_2 | X_1, U) = 0$). Furthermore, $U$
    is still uniform over two values since $I(X_1; U) = 0$. Therefore, $X_2$ is uniform over either one or two values.
    The case of one value is impossible because $H(U | X_1, X_2) = 0$.
    A symmetric argument is used for $X_1$ when $X_2$ is fixed, proving the statement -- any vertex 
    which has a nonzero probability of occurring has exactly two neighbors. For the ``if'' direction, since this is a bipartite graph, all cycles must be
    of even length. Therefore, we can use two colors to color all the edges such that no two edges which have a common vertex are of the same color.
    Taking $U$ to represent this coloring of edges satisfies the above predicate.
\end{proof}
Furthermore, this predicate is satisfiable -- for any finite collection of even-length cycles, we can clearly find $X_1, X_2$
such that their characteristic bipartite graph consists exactly of this collection of cycles.

\subsection{Overview of the construction}
We now give an overview of the following steps of the construction.
Section~\ref{section:grid} defines the predicate $\pTori'(X^2, Y^2, Z)$, which enforces that the characteristic bipartite graph of $X_1$ and $X_2$ is a collection of cycles, similarly for $Y_1$ and $Y_2$. Finally, we require that $Z$ be distributed uniformly over two values and that the three variables $X^2, Y^2, Z$ be independent. The distribution of $(X^2, Y^2)$ then represents a collection of tori, with each quadruple of values of $(X_1, X_2, Y_1, Y_2)$ representing a vertex in some torus. The addition of variable $Z$ effectively creates a corresponding copy of this collection.

With the goal of creating meaningful labels to be applied to the vertices of the aforementioned graph, which are represented by $k$-tuple of binary variables $W^k$, Section~\ref{section:labels} defines the predicates $\pSw(W^k, V^k, \bar{V}^k, F)$ and $\pCol'(W^k, V^k, \bar{V}^k, F)$. The former ensures that $V_i = (1 - W_i)F, \bar{V}_i = W_i F$ (up to relabeling) with the side-effect of requiring each $W_i \sim \text{Bern}(\frac12)$. The latter predicate restricts $W^k$ such that the only values that are possible have either $W_k = 1$ and exactly one $W_i = 0$, or $W_k = 0$ and exactly one $W_i = 1$, for some $i \in \{1, \ldots, k - 1\}$.

Section~\ref{section:edges} combines the prior predicates in predicate $\pCtori'(X^2, Y^2, Z, W^k, V^k, \bar{V}^k, F)$ in such a way that each vertex is assigned exactly one label, i. e. $W^k$ depends functionally on $(X^2, Y^2, Z)$.
This is extended in predicate $\pOtori'(X^2, Y^2, Z, W^k, V^k, \bar{V}^k, F)$, which assigns four possible \emph{groups} to vertex labels and enforces a structure as shown in Figure~\ref{fig:torus}.

Finally, the predicate $\pTtori'(X^2, Y^2, Z, W^k, V^k, \bar{V}^k, F)$ defined in Section~\ref{section:final} restricts the possible labels of vertices connected by edges so as to enforce the given vertical and horizontal constraints of the given tile system.

\subsection{Grid}\label{section:grid}
A collection of tori is constructed by Li using the following predicate
(where $X^2$ represents a pair of variables $(X_1, X_2)$, similarly for $Y^2$):
\[
\pTori(X^2, Y^2) : \pCycs(X^2) \land \pCycs(Y^2) \land I(X^2; Y^2) = 0, 
\]

When the above holds, fixing any three of the variables $X_1, X_2, Y_1, Y_2$ leaves two possible values of the fourth.
Similarly, fixing one variable from $X_1, X_2$ and one from $Y_1, Y_2$ gives the remaining two variables a distribution over four values. This is visualized by a graph similar
in idea to the bipartite characteristic graph: its vertices are quadruples of values $(x_1, x_2, y_1, y_2)$ which satisfy $\mathbf{P}(X_1 = x_1 \land X_2 = x_2 \land Y_1 = y_1 \land Y_2 = y_2) > 0$, with edges connecting any two quadruples which differ in exactly one out of these four values. When arranged in a grid with possible
pairs $(X_1, X_2)$ on one axis and $(Y_1, Y_2)$ on the other, as in Figure~\ref{fig:torus}, the torus structure becomes apparent.
Again, this predicate is satisfiable in the sense that any collection of tori which is a product of two collections
of even-length cycles has a representation by random variables $X^2, Y^2$.

Our construction departs slightly from the construction of Li, adding another coordinate $Z$, corresponding to taking two copies
of the collection of tori, with the edges and faces described above preserved when $Z$ is fixed. Additionally, fixing
$X^2$ and $Y^2$ but not Z ``connects'' two corresponding vertices in the two copies. The predicate for this is as follows:
\begin{align*}
    \pTori'(X^2, Y^2, Z) : \ & \pCycs(X^2) \land \pCycs(Y^2) \land \pUnif_2(Z) \\
     & \land I(X^2, Y^2; Z) = 0 \land I(X^2, Z; Y^2) = 0 \land I(Y^2, Z; X^2) = 0.
\end{align*}

\begin{figure}[t]
    \centering
    \begin{tikzpicture}[
        every node/.append style={scale=0.8, transform shape},
        dot/.style={outer sep=-1pt,circle,minimum size=3pt,draw,
        append after command={\pgfextra \draw (\tikzlastnode) ; \endpgfextra}}]
            \matrix (m) [matrix of nodes,
                         nodes in empty cells,
                         nodes={anchor=center},
                         column sep={1.1cm, between origins},
                         row sep={1.1cm, between origins}]{
                $y_1, \tilde{y}_1$ & |[dot]| & |[dot]| & |[dot]| & |[dot]| & |[dot]| & |[dot]| & |[dot]| & |[dot]| & |[dot]| & |[dot]| \\
                $y_2, \tilde{y}_1$ & |[dot]| & |[dot]| & |[dot]| & |[dot]| & |[dot]| & |[dot]| & |[dot]| & |[dot]| & |[dot]| & |[dot]| \\
                $y_2, \tilde{y}_2$ & |[dot]| & |[dot]| & |[dot]| & |[dot]| & |[dot]| & |[dot]| & |[dot]| & |[dot]| & |[dot]| & |[dot]| \\
                $y_1, \tilde{y}_2$ & |[dot]| & |[dot]| & |[dot]| & |[dot]| & |[dot]| & |[dot]| & |[dot]| & |[dot]| & |[dot]| & |[dot]| \\
                                     & $x_1, \tilde{x}_1$ & $x_2, \tilde{x}_1$ & $x_2, \tilde{x}_2$ & $x_3, \tilde{x}_2$ & $x_3, \tilde{x}_3$ & $x_3, \tilde{x}_1$ & $x_4, \tilde{x}_4$ & $x_5, \tilde{x}_4$ & $x_5, \tilde{x}_5$ & $x_5, \tilde{x}_4$ \\
            };
            \foreach \x [evaluate=\x as \xx using {int(\x+1)}] in {8,...,10}{
                \foreach \y [evaluate=\y as \yy using {int(\y+1)},
                             evaluate=\y as \face using {int(Mod(\x,2)+2*Mod(\y,2))},
                             evaluate=\y as \first using {int(2-Mod(\face,2))},
                             evaluate=\y as \second using {int(1+int(\face/2))}] in {1,...,3}{
                    \node (face-\y-\x) at ($(m-\y-\x)!0.5!(m-\yy-\xx)$) {\first\second};
                }
            }

            \foreach \x [evaluate=\x as \xx using {int(\x+1)}] in {8,...,11}{
                \foreach \y [evaluate=\y as \yy using {int(\y+1)},
                             evaluate=\y as \face using {int(Mod(\x+\y+1,2) + 2*pow(-1,\y+1) + 2*Mod(\y+1,2))}] in {1,...,4} {
                    \node (vertex-\y-\x) at (m-\y-\x) {\face};
                }
            }

            \foreach \y [evaluate=\y as \yy using {int(\y+1)}] in {1,...,3}{
                \draw (m-\y-1) edge (m-\yy-1);
                \draw (m-4-1) edge[bend left=45, looseness=0.7,line width=0.25pt] (m-1-1);
            }

            \foreach \x in {2,...,11}{
                \foreach \y [evaluate=\y as \yy using {int(\y+1)}] in {1,...,3}{
                    \draw (m-\y-\x) edge (m-\yy-\x);
                    \draw (m-4-\x) edge[bend left=45, looseness=0.7, color=lightgray] (m-1-\x);
                }
            }

            \foreach \x [evaluate=\x as \xx using {int(\x+1)}] in {2,...,6}{
                \draw (m-5-\x) edge (m-5-\xx);
                \draw (m-5-7) edge[looseness=0.5,bend right=30,line width=0.25pt] (m-5-2);
            }
            \foreach \x [evaluate=\x as \xx using {int(\x+1)}] in {8,...,10}{
                \draw (m-5-\x) edge (m-5-\xx);
                \draw (m-5-11) edge[looseness=0.5,bend right=30,line width=0.25pt] (m-5-8);
            }

            \foreach \y in {1,...,4}{
                \foreach \x [evaluate=\x as \xx using {int(\x+1)}] in {2,...,6}{
                    \draw (m-\y-\x) edge (m-\y-\xx);
                    \draw (m-\y-7) edge[looseness=0.5,bend right=30,color=lightgray] (m-\y-2);
                }
                \foreach \x [evaluate=\x as \xx using {int(\x+1)}] in {8,...,10}{
                    \draw (m-\y-\x) edge (m-\y-\xx);
                    \draw (m-\y-11) edge[looseness=0.5,bend right=30,color=lightgray] (m-\y-8);
                }
            }

            \node[draw=gray,fill,fill opacity=0.1,fit=(m-1-4)(m-2-4),rounded corners,inner sep=8pt,label={fix all but $Y_1$}] {};
            \node[draw=gray,fill,fill opacity=0.1,fit=(m-2-4)(m-3-4),rounded corners,inner sep=8pt,label=below:{fix all but $Y_2$}] {};
            \node[draw=gray,fill,fill opacity=0.1,fit=(m-2-3)(m-2-4),rounded corners,inner sep=8pt,label=above left:{fix all but $X_2$}] {};
            \node[draw=gray,fill,fill opacity=0.1,fit=(m-2-4)(m-2-5),rounded corners,inner sep=8pt,label=above right:{fix all but $X_1$}] {};
            \node[draw=gray,fill,fill opacity=0.1,fit=(m-3-6)(m-4-7),rounded corners,inner sep=8pt,label=above:{fix $(Z, X_1, Y_2)$}] {};
    \end{tikzpicture}
    \caption{
        Visualization of the tori which are a product of the cycles created by $(X_1, X_2)$ and $(Y_1, Y_2)$,
        with each vertex corresponding to a quadruple of the values of $(X_1, X_2, Y_1, Y_2)$. The axes show these
        cycles --- a quadruple's $(X_1, X_2)$ (resp. $(Y_1, Y_2)$) values are determined by projecting onto the horizontal
        (resp. vertical) axis.
        Additionally, the left torus shows highlighted edges which arise when all but one variable (of $X_1$, $X_2$, $Y_1$, $Y_2$, $Z$) are fixed as well
        as an example face which arises when $Z$ and two other variables are fixed.
        The right torus has each face labeled with its type and each vertex labeled with its group -- these are used in Section~\ref{section:tiles} in order to restrict allowed labelings of the vertices.
        The second ``corresponding'' torus and the $Z$ axis are omitted for clarity.
    }
    \label{fig:torus}
\end{figure}
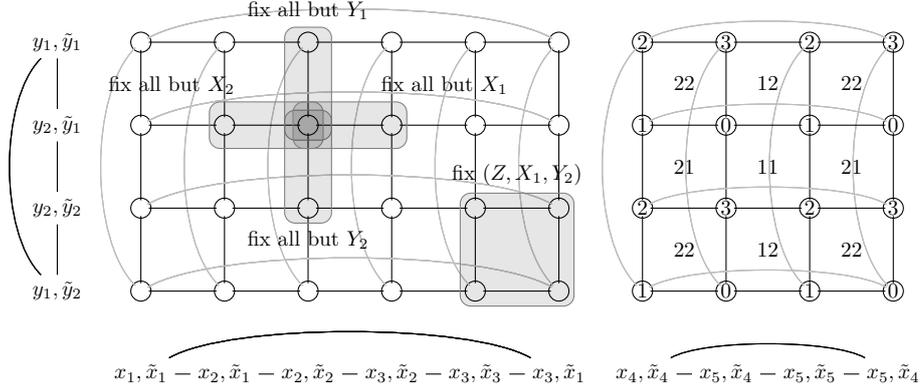
\subsection{Vertex labels}\label{section:labels}
The basis for constructing labels which will later on be assigned to vertices is the following predicate defined by Li:
\begin{align*}
 \pFlip(F, G_1, G_2): \ & \exists U \leq_i 4, Z^2 \leq_i 3: \pUnif_4(U) \land \pUnif_2(F) \\
 & \land H(F, G_1, G_2 | U) = I(G_1; G_2 | F) = 0 \\
 & \land \pUnif_3(Z_1) \land I(Z_1; G_1) = H(U | G_1, Z_1) = 0 \\
 & \land \pUnif_3(Z_2) \land I(Z_2; G_2) = H(U | G_2, Z_2) = 0.
\end{align*}
Recalling that $\text{Unif}(S)$ denotes a uniform distribution over set $S$, the following is shown:
\begin{lemma} \label{lemma:flip1}
    $\pFlip(F, G_1, G_2)$ is satisfied iff, up to relabeling, $(F, G_1, G_2)$ has the distribution
    ${\text{Unif}(\{(0, 0, 0), (0, 1, 0), (1, 0, 0), (1, 0, 1)\})}$.
\end{lemma}
\begin{proof}
    For the ``if'' direction, we let $U$ be functionally dependent on $(F, G_1, G_2)$ and take a different value for each
    of the four values in the distribution of $(F, G_1, G_2)$. Letting $Z_1$ be independent from $G_1$, we can map $(G_1, Z_1)$
    to $U$: when $G_1 = 1$, we always give $U$ corresponding to $(F, G_1, G_2) = (0, 1, 0)$ and when $G_1 = 0$, the three
    values of $Z_1$ correspond to the three values $(0, 0, 0)$, $(1, 0, 0)$, $(1, 0, 1)$ of $(F, G_1, G_2)$. A similar
    argument applies to $G_2, Z_2$, satisfying the predicate.

    For the ``only if'' direction, we first show that we have either ${G_1 \sim \text{Unif}(\{1, 2, 3, 4\})}$ or ${G_1 \sim \text{Bern}(\frac14)}$. Consider the
    functional dependency $(G_1, Z_1) \mapsto U$. Since $G_1$ and $Z_1$ are independent and $Z_1$ is uniform over three values,
    we have $\textbf{P}(G_1 = g \land Z_1 = z) = \frac13 \textbf{P}(G_1 = g)$. Since ${H(G_1 | U) = 0}$, we cannot have
    (nonzero-probability) $g \neq g'$ and $z, z'$ such that $(g, z)$ and $(g', z')$ map to the same value of $U$.
    Hence, for each $g$ with nonzero probability, the three values of $Z_1$ map to either three different values or a
    single one. If the former occurs, then exactly one other value of $G_1$ can have nonzero probability and must map to
    a single value of $U$, giving $G_1 \sim \text{Bern}(\frac14)$. Otherwise, all values of $G_1$ must have nonzero probability
    and map to the same value of $U$ regardless of $Z_1$, thus $G_1 \sim \text{Unif}(\{1, 2, 3, 4\})$.

    The above proof applies analogously to $G_2$. We now strengthen this to state that $G_i \sim {\text{Bern}(\frac14)}$ for both $i \in \{1, 2\}$.
    Suppose otherwise -- then for some $f$, both $G_1|F=f$ and $G_2|F=f$ take two values, contradicting $I(G_1; G_2 | F) = 0$. This means that (up to relabeling)
    \[
        (G_1, G_2) = \begin{cases}
            (0, 0) & \text{ when } U \in \{1, 2\}, \\
            (0, 1) & \text{ when } U = 3, \\
            (1, 0) & \text{ when } U = 4.
        \end{cases}
    \]
    Only the mapping of $U$ to $F$ remains. If its values are different for $U = 1$ and $U = 2$, then this is the desired
    distribution. Otherwise, we have $0 = I(G_1; G_2 | F) = H(G_1|F) - H(G_1|G_2,F) = H(G_1|F) > 0$, contradiction.
    This leaves the desired distribution as the only possible solution, and it clearly satisfies the predicate.
\end{proof}
For any $k \geq 4$, Li defines the predicate $\pSw$ which allows us to represent strings of $k$ bits.
Here, $W^k$ represents a tuple of variables $(W_1, \ldots, W_k)$, same for $V^k, \bar{V}^k$:
Intuitively, the value of $W_i$ determines whether $F$ is copied into $V_i$ or $\bar{V}_i$, hence the name $\pSw$ for \emph{switch}.
\begin{align*}
 \pSw(W^k, V^k, \bar{V}^k, F): \ & \exists G \leq_i 2:  I(W^k; F, G) = 0 \\
 & \land \bigwedge_{i \in \{1, \ldots, k\}} \big( \pUnif_2(W_i) \land  H(V_i, \bar{V}_i | W_i, F) = I(V_i; \bar{V}_i | W_i) = 0 \\
 & \qquad \qquad \land \pFlip(F, G, V_i) \land \pFlip(F, G, \bar{V}_i) \big)
\end{align*}
\begin{lemma}
    If $\pSw(W^k, V^k, \bar{V}^k, F)$ is satisfied, then we have (without loss of generality) $V_i = {(1-W_i)F}, \bar{V}_i = {W_i F}$
    for all $i$.
\end{lemma}
\begin{proof}
    Consider any $i$. Since $\pFlip(F, G, V_i)$ holds, we can assume $F \sim {\text{Bern}(\frac12)}, G|F \sim {\text{Bern}(\frac{1-F}{2})}, V_i|F \sim \text{Bern}(\frac{F}{2})$.
    Since $W_i$ and $F$ are independent and both uniform over $\{0, 1\}$, we have $(W_i, F) \sim \text{Unif}(\{0, 1\}^2)$.
    From $V_i|F \sim \text{Bern}(\frac{F}{2})$ and $H(V_i|W_i, F) = 0$, we can essentially only have $V_i = F, V_i = W_i F$ or
    $V_i = (1-W_i)F$; the same for $\bar{V}_i$. The first case would contradict $I(V_i; \bar{V}_i | F) = 0$ and the other
    two options are symmetrical, therefore we assume $V_i = (1-W_i)F$ without loss of generality. Then $\bar{V}_i = (1-W_i)F$
    would contradict $I(V_i; \bar{V}_i | F) = 0$, therefore we assume $\bar{V}_i = W_iF$.
\end{proof}
The predicate $\pSw$ is satisfiable: we let $(F, G)$ take each of the values $(0, 0)$, $(0, 1)$ with probability $\frac14$, in which case we let $V_i = \bar{V}_i = 0$, and the value $(1, 0)$ with probability $\frac12$, in which case $V_i = 1 - W_i, \bar{V}_i = W_i$. As long as each $W_i \sim \text{Bern}(\frac12)$
(recall that $\text{Bern}(p)$ denotes a Bernoulli distribution with parameter $p$), we can satisfy the predicate for any distribution of $W^k$. This predicate additionally has the following property:
\begin{lemma}
    If $\pSw(W^k, V^k, \bar{V}^k, F)$ is satisfied, then for any $S, \bar{S} \subseteq \{1, \ldots, k\}$ we have
    \[
        H(F|V_S, \bar{V}_{\bar{S}}, W^k) = \mathbf{P}(\text{sat}(W^k, S, \bar{S}) = 1),
    \]
    where
    \[
        \text{sat}(w^k, S, \bar{S}) = \Big( \prod_{i \in S} w_i \Big) \Big( \prod_{i \in \bar{S}} (1 - w_i) \Big),
    \]
    i.e. if $w_i = 1$ for all $i \in S$ and $w_i = 0$ for $i \in \bar{S}$ then $\text{sat}(w^k, S, \bar{S})$ equals 1
    and otherwise it equals 0.
\end{lemma}
\begin{proof}
    Without loss of generality, we have $V_i = {(1 - W_i)F}, \bar{V}_i = {W_i F}$ for all $i \in \{1, \ldots, k\}$.
    Since $I(W^k; F) = 0$, knowing the value of $W^k$ gives no additional information about $F$. However, if we also know
    the value of $V_i = (1-W_i)F$ for some $i$ such that $W_i = 0$, then the value of $F$ is revealed. Conversely, if we only
    know the values of $V_i$ for $i$ such that $W_i = 1$, then $F$ conditioned on this knowledge still follows $\text{Bern}(\frac12)$.
    An analogous argument can be used for $\bar{V}_i$, and these combined give the result that for any
    $w^k \in \{0, 1\}^k$ and $S, \bar{S} \subseteq \{1, \ldots, k\}$, we have
    \[
        H(F|V_S, \bar{V}_{\bar{S}}, W^k = w^k) = \text{sat}(w^k, S, \bar{S}).
    \]
    Summing over all possible values of $w^k$, we get
    \[
        H(F|V_S, V_{\bar{S}}, W^k) = \mathbf{E}[\text{sat}(w^k, S, \bar{S})] = \mathbf{P}(\text{sat}(w^k, S, \bar{S}) = 1). \qedhere
    \]
\end{proof}
This immediately yields the following:
\begin{lemma} \label{cor:sat}
    The equality $H(F|V_{\{i \in \{1, \ldots, k \} \; : \; w_i = 1\}}, \bar{V}_{\{i \in \{1, \ldots, k \} \; : \; w_i = 0\}}, W^k) = \mathbf{P}(W^k = w^k)$ holds for any $w^k \in \{0, 1\}^k$. \qed
\end{lemma}

Using these properties, we can disallow certain values of $W^k$ from occurring using a CI statement. This is used by Li
to limit the possible values of $W^k$ to the set $T_k \subseteq \{0, 1\}^k$, which consists of $2(k-1)$ \emph{labels}
(referred to as \emph{colors} by Li), each one with a \emph{value} and \emph{sign}. The set consists of strings in which exactly one bit differs from the last bit,
that is for any $w^k \in T_k$ we have $w_j \neq w_k$ for exactly one $j \neq k$.
The sign of the label is determined by $w_k$ --- negative when $w_k=1$, positive when $w_k=0$ --- and $j$ is the value of the label.
For example, the elements of $T_4 = \{0111, 1011, 1101, 1000, 0100, 0010\}$ correspond in order to labels $\{-1, -2, -3, +1, +2, +3\}$.

Li's predicate for enforcing this is simple:
\begin{align*}
    \pCol(W^k, V^k, \bar{V}^k, F) : \ & \ \pSw(W^k, V^k, \bar{V}^k, F) \\
                                         & \land \bigwedge_{w^k \in \{0,1\}^k \setminus T_k} \left( H(F|V_{\{i:w_i = 1\}}, \bar{V}_{\{i:w_i=0\}}, W^k) = 0\right),
\end{align*}
This predicate simply disallows the occurrence of any $w^k \not \in T_k$, hence we have
\begin{lemma}
    If $\pCol(W^k, V^k, \bar{V}^k, F)$ is satisfied, then $\mathbf{P}(W^k \not \in T_k) = 0$. \qed
\end{lemma}
While sufficient for showing undecidability, this predicate
cannot be used in a polynomial-time reduction due to its size being exponential with regard to $k$. However, an equivalent
polynomial-size predicate can be easily constructed: with
\[
    \bigwedge_{\substack{i, j \in \{1, \ldots, k-1\}, \\ i < j}} (H(F|V_{\{i, j\}}, \bar{V}_{\{k\}}, W^k) = 0),
\]
we disallow all $w^k$ such that $w_k = 0$ and $w_i = w_j = 1$ for some $1 \leq i < j < k$.  Similarly,
\[
    \bigwedge_{\substack{i, j \in \{1, \ldots, k-1\}, \\ i < j}} (H(F|V_{\{k\}}, \bar{V}_{\{i, j\}}, W^k) = 0)
\]
disallows all $w^k$ such that $w_k = 1$ and $w_i = w_j = 0$ for some $1 \leq i < j < k$. The only remaining $w^k$
have either $w_k = 0$ and at most one 1 in $\{1, \ldots, k-1\}$ or have $w_k = 1$ and at most one $0$ in $\{1, \ldots, k-1\}$.
The strings $0^k$ and $1^k$ are disallowed by $H(F|V_{\{1, \ldots, k\}}, \bar{V}_{\varnothing}, W^k) = 0$ and $H(F|V_{\varnothing}, \bar{V}_{\{1, \ldots, k\}}, W^k) = 0$.
Combined, these yield the polynomial-size predicate
\begin{align*}
    \pCol'(W^k, V^k, \bar{V}^k, F) & : \pSw(W^k, V^k, \bar{V}^k, F) \\
        & \land \bigwedge_{\substack{i, j \in \{1, \ldots, k-1\}, \\ i < j}} (H(F|V_{\{i, j\}}, \bar{V}_{\{k\}}, W^k) = 0) \\
        & \land \bigwedge_{\substack{i, j \in \{1, \ldots, k-1\}, \\ i < j}} (H(F|V_{\{k\}}, \bar{V}_{\{i, j\}}, W^k) = 0) \\
        & \land H(F|V_{\{1, \ldots, k\}}, \bar{V}_{\varnothing}, W^k) = 0 \\
        & \land H(F|V_{\varnothing}, \bar{V}_{\{1, \ldots, k\}}, W^k) = 0
\end{align*}
equivalent to the exponential-size $\pCol$.

\subsection{Edge constraints}\label{section:edges}
The next predicate is the following, with $\pCol$ used in place of $\pCol'$ in Li's original construction:
\begin{align*}
 \pCold(X, W^k, V^k, \bar{V}^k, F): \ & \pCol'(W^k, V^k, \bar{V}^k, F) \\
 & \land H(W^k | X) = I(V^k, \bar{V}^k, F; X | W^k) = 0.
\end{align*}
This predicate will represent the labeling of vertices, with $X$ representing
the coordinate and $W^k$ (which depends functionally on $X$) representing its label.  For any $x \in \mathcal{X}$,
denote by $w^k(x)$ the unique value of $w^k$ which satisfies $\mathbf{P}(W^k = w^k | X = x) > 0$.

Suppose that $\pCold(X, W^k, V^k, \bar{V}^k, F)$ is satisfied and let $E$ be any random variable
which \emph{splits $X$ into sets of (a constant) size $l$} -- we say this is the case when $H(E | X) = 0$ and $X | E = e$
is uniform over $l$ values for all $e \in \mathcal{E}$.
This is not verified by a predicate; rather, we will only choose
$E$ which have this property by definition. Fixing subsets of indices $S, \bar{S} \subseteq \{1, \ldots, k\}$, we define for any
$e \in \mathcal{E}$ the value $a_e$:
\[
    a_e = | \{ x : \mathbf{P}(X = x | E = e) > 0 \land \text{sat}(w^k(x), S, \bar{S}) = 1 \} |
\]
In order to impose restrictions on the possible values of $a_e$, Li defines the following predicates:
\begingroup
\allowdisplaybreaks
\begin{align*}
    \pSat_{\neq 1/2, S, \bar{S}}(E, W^k, V^k, \bar{V}^k, F): \ & \exists U \leq_i 2: \pUnif_2(U) \land I(U; E, V_{S}, \bar{V}_{\bar{S}}) = 0 \\
 & \land H(F | V_{S}, \bar{V}_{\bar{S}}, E, U) = 0,  \\[6pt]
 \pSat_{\le 1/2, S, \bar{S}}(E, W^k, V^k, \bar{V}^k, F):
 & \exists U \leq_i 3: \pUnif_3(U) \land I(U; E, V_{S}, \bar{V}_{\bar{S}}) = 0 \\
 & \land H(F | V_{S}, \bar{V}_{\bar{S}}, E, U) = 0,  \\[6pt]
 \pSat_{\le 3/4, S, \bar{S}}(E, W^k, V^k, \bar{V}^k, F):
 & \exists U \leq_i 105: \pUnif_{105}(U) \land I(U; E, V_{S}, \bar{V}_{\bar{S}}) = 0 \\
 & \land H(F | V_{S}, \bar{V}_{\bar{S}}, E, U) = 0.
\end{align*}
\endgroup
The predicates satisfy the following properties.
\begin{lemma}\label{lemma:cold1}
    If $\pCold(X, W^k, V^k, \bar{V}^k, F)$ is satisfied and $E$ splits $X$ into sets of size 2, then $\pSat_{\neq 1/2, S, \bar{S}}(E, W^k, V^k, \bar{V}^k, F)$
    is satisfied
    iff $a_e \neq 1$ for all $e \in \mathcal{E}$.
\end{lemma}
\begin{lemma}\label{lemma:cold2}
    If $\pCold(X, W^k, V^k, \bar{V}^k, F)$ is satisfied and $E$ splits $X$ into sets of size 2, then $\pSat_{\leq 1/2, S, \bar{S}}(E, W^k, V^k, \bar{V}^k, F)$
    is satisfied
    iff $a_e \leq 1$ for all $e \in \mathcal{E}$.
\end{lemma}
\begin{lemma}\label{lemma:cold3}
    If $\pCold(X, W^k, V^k, \bar{V}^k, F)$ is satisfied and $E$ splits $X$ into sets of size 4, then $\pSat_{\leq 3/4, S, \bar{S}}(E, W^k, V^k, \bar{V}^k, F)$
    is satisfied
    iff $a_e \leq 3$ for all $e \in \mathcal{E}$.
\end{lemma}
\begin{proof}[Proof of Lemma~\ref{lemma:cold1}]
    Consider a fixed $e \in \mathcal{E}$. Denote by $\vec{0}$ a vector whose elements are all zero. Whenever $V_S \neq \vec{0}$ or $\bar{V}_{\bar{S}} \neq \vec{0}$, then $F$ must be 1
    and so the functional dependency $H(F | V_{S}, \bar{V}_{\bar{S}}, E, U) = 0$ is satisfied regardless of the value of $U$.
    For $V_S = \vec{0}$, $\bar{V}_{\bar{S}} = \vec{0}$ we have
    \begin{align*}
        & \mathbf{P}(V_S = \vec{0} \land \bar{V}_{\bar{S}} = \vec{0} \land F = 1 \land E = e) = \mathbf{P}(F = 1) \mathbf{P}(E = e \land \text{sat}(w^k(X), S, \bar{S}) = 1), \\
        & \mathbf{P}(V_S = \vec{0} \land \bar{V}_{\bar{S}} = \vec{0} \land F = 0 \land E = e) = \mathbf{P}(F = 0) \mathbf{P}(E = e), \\
        & \text{hence } \mathbf{P}(F = 1 | V_S = \vec{0} \land \bar{V}_{\bar{S}} = \vec{0} \land E = e) = \frac{\frac12 \cdot \frac{a_e}{l}}{\frac12 + \frac12 \cdot \frac{a_e}{l}} = \frac{a_e}{l + a_e}.
    \end{align*}
    These equalities are based on the fact that ${V_i = (1 - W_i)F}$.
    We conclude that $F|(V_S = \vec{0}, \bar{V}_{\bar{S}} = \vec{0}, E = e) \sim \text{Bern}(\frac{a_e}{l + a_e})$.

    When $a_e = 0$, the value of $F$ is determined only by whether any of $V_S, \bar{V}_{\bar{S}}$ is nonzero and the value of $U$ is irrelevant. When $a_e = 2$, we can set $U = F$. On the other hand,
    ${F|(V_S = \vec{0}, \bar{V}_{S} = \vec{0}, E=e)} \sim \text{Bern}(\frac13)$ when $a_e = 1$. Conditioned on $V_S = \vec{0}, \bar{V}_{S} = \vec{0}, E=e$, the value of
    $U$ functionally determines $F$ --- this is impossible with $U \sim \text{Bern}(\frac12), F \sim \text{Bern}(\frac13)$.
\end{proof}
\begin{proof}[Proof of Lemma~\ref{lemma:cold2}]
    The proof is analogous to the above. We have $H(F|V_S, \bar{V}_{\bar{S}}, E, U) = 0$ when $V_S \neq \vec{0}$ or $\bar{V}_{\bar{S}} \neq \vec{0}$
    and also $F|(V_S = \vec{0}, \bar{V}_{\bar{S}} = \vec{0}, E = e) \sim \text{Bern}(\frac{a_e}{l + a_e})$. We cannot have
    $U \sim \text{Unif}(\{1, 2, 3\})$ map to $\text{Bern}(\frac12)$ when $a_e = 2$. When $a_e = 0$, the value of $U$ is irrelevant
    and when $a_e = 1$, then $F|(V_S = \vec{0}, \bar{V}_{\bar{S}} = \vec{0}, E = e) \sim \text{Bern}(\frac{1}{3})$ and so
    we can satisfy the functional dependence $H(F|V_S, \bar{V}_{\bar{S}}, E, U) = 0$.
\end{proof}
\begin{proof}[Proof of Lemma~\ref{lemma:cold3}]
    Again, we have $F|(V_S = \vec{0}, \bar{V}_{\bar{S}} = \vec{0}, E = e) \sim \text{Bern}(\frac{a_e}{l + a_e})$, where $a_e$
    is the number of of values of $X$ in this set which satisfy ${\text{sat}(W^k, S, \bar{S}) = 1}$.
    Since 105 is divisible by 3, 5 and 7, it is possible for $U \sim \text{Unif}(\{1, \ldots, 105\})$ to map to $\text{Bern}(\frac{a_e}{l+a_e})$ whenever $a_e \in \{0, 1, 2, 3\}$, but not when $a_e = 4$.
\end{proof}

The next defined predicate is the following:
\begin{align*}
    \pCtori'(X^2, Y^2, Z, W^k, V^k, \bar{V}^k, F) & \ : \pTori'(X^2, Y^2, Z) \\
 & \land \pCold((X^2, Y^2, Z), W^k, V^k, \bar{V}^k, F),
\end{align*}
which simply implies applying labels (without any constraints) to the vertices of the tori. In the original predicate,
which uses $\pTori$ instead of $\pTori'$, there is no coordinate $Z$.
Clearly, this predicate is satisfiable in the sense that any collection of pairs of tori of even size which is labeled
in a manner that satisfies the requirement $W_i \sim \text{Bern}(\frac12)$ has a corresponding representation by $X^2, Y^2, Z, W^k, V^k, \bar{V}^k, F$.
In particular, this $W_i$ requirement is satisfied by any labeling in which any pair of corresponding vertices (those which
differ only in the $Z$ coordinate) has labels of the same value but opposite sign. This is because negating the sign of a label
corresponds to negating all of its bits.
For a set of labels $\mathcal{L} = \{0, \ldots, l-1\}$, we now label the vertices with labels from the set $\{0, \ldots, 4l - 1\}$
and so we set $k = 4l + 1$.
For any $i \in \{0, \ldots, l\}$, $j \in \{0, \ldots, 3\}$, we identify all four labels $4i+j$ with the original label $i$,
referring to any vertex whose label is $4i + j$ as a \emph{group $j$ vertex}. The group of a vertex is used to orient
it relative to its neighbors -- this is achieved by the following predicate:
\begin{alignat*}{3}
 & \mathrlap{ \pOtori'(X^2, Y^2, Z, W^k, V^k, \bar{V}^k, F) : } \\
 & \ \mathrlap{ \pCtori'(X^2, Y^2, Z, W^k, V^k, \bar{V}^k, F) } \\
 & \land \ \mathrlap{ \pSat_{\neq 1/2,  \{k\},  \varnothing}((X_1, X_2, Y_1, Z), W^k, V^k, \bar{V}^k, F) } \\
 & \land \ \mathrlap{ \pSat_{\neq 1/2,  \{k\},  \varnothing}((X_1, X_2, Y_2, Z), W^k, V^k, \bar{V}^k, F) } \\
 & \land \ \mathrlap{ \pSat_{\neq 1/2,  \{k\},  \varnothing}((X_1, Y_1, Y_2, Z), W^k, V^k, \bar{V}^k, F) } \\
 & \land \ \mathrlap{ \pSat_{\neq 1/2,  \{k\},  \varnothing}((X_2, Y_1, Y_2, Z), W^k, V^k, \bar{V}^k, F) } \\
 & \land \smash{\bigwedge_{j_1, j_2 \in J_1}} \big( && \pSat_{\le 1/2,  \varnothing,  \{1, \ldots, k\} \setminus\{j_1, j_2\}}((X_1, X_2, Y_1, Z), W^k, V^k, \bar{V}^k, F) \\
 & && \land \pSat_{\le 1/2,  \varnothing,  \{1, \ldots, k\} \setminus\{j_1, j_2\}}((X_1, X_2, Y_2, Z), W^k, V^k, \bar{V}^k, F) \\
 & && \land \pSat_{\le 1/2,  \{1, \ldots, k\} \setminus\{j_1, j_2\},  \varnothing}((X_1, X_2, Y_1, Z), W^k, V^k, \bar{V}^k, F) \\
 & && \land \pSat_{\le 1/2,  \{1, \ldots, k\} \setminus\{j_1, j_2\},  \varnothing}((X_1, X_2, Y_2, Z), W^k, V^k, \bar{V}^k, F)\big) \\
 & \land \smash{\bigwedge_{j_1, j_2 \in J_2}} \big( && \pSat_{\le 1/2,  \varnothing,  \{1, \ldots, k\} \setminus\{j_1, j_2\}}((X_1, Y_1, Y_2, Z), W^k, V^k, \bar{V}^k, F) \\
 & && \land \ \pSat_{\le 1/2,  \varnothing,  \{1, \ldots, k\} \setminus\{j_1, j_2\}}((X_2, Y_1, Y_2, Z), W^k, V^k, \bar{V}^k, F) \\
 & && \land \ \pSat_{\le 1/2,  \{1, \ldots, k\} \setminus\{j_1, j_2\},  \varnothing}((X_1, Y_1, Y_2, Z), W^k, V^k, \bar{V}^k, F) \\
 & && \land \ \pSat_{\le 1/2,  \{1, \ldots, k\} \setminus\{j_1, j_2\},  \varnothing}((X_2, Y_1, Y_2, Z), W^k, V^k, \bar{V}^k, F)\big) \\
 & \land \ \mathrlap{\pSat_{\leq 1/2, \{k\}, \varnothing}((X_1, X_2, Y_1, Y_2), W^k, V^k, \bar{V}^k, F)} \\
 & \land \ \mathrlap{\pSat_{\leq 1/2, \varnothing, \{k\}}((X_1, X_2, Y_1, Y_2), W^k, V^k, \bar{V}^k, F),}
\end{alignat*}
where
\begin{align*}
    J_1 & = \{(j_1, j_2) \in \{1, \ldots, k-1\}: \{j_1 \ \text{mod} \ 4,  j_2 \ \text{mod} \ 4\} \notin \{\{0, 1\}, \{2, 3\}\} \}, \\
    J_2 & = \{(j_1, j_2) \in \{1, \ldots, k-1\}: \{j_1 \ \text{mod} \ 4,  j_2 \ \text{mod} \ 4\} \notin \{\{1, 2\}, \{0, 3\}\} \}.
\end{align*}
The only difference between $\pOtori'$ and Li's original $\pOtori$ is the added variable $Z$ and the final two $\pSat_{\leq 1/2}$ predicates. We have the following fact:
\begin{lemma}\label{lemma:statements}
    If $\pOtori\; '(X^2, Y^2,Z , W^k, V^k, \bar{V}^k, F)$ is satisfied, then the following statements hold:
    \begin{enumerate}
        \item within each torus, all vertices' labels have the same sign; \label{otori:sign}
        \item any two vertices differing only in the $Z$ coordinate have opposite sign; \label{otori:3d}
        \item any pair of vertices connected by a vertical edge either has groups 1 and 0 or 2 and 3; \label{otori:vertical}
        \item any pair of vertices connected by a horizontal edge either has groups 1 and 2 or 3 and 0. \label{otori:horizontal}
    \end{enumerate}
\end{lemma}
\begin{proof}
    Recall that fixing the variable $Z$ along with any three of the variables $X_1, X_2, Y_1, Y_2$ leaves two possible values for the remaining variable.
    These correspond to two vertices of an edge, as illustrated in Figure~\ref{fig:torus}.
    Therefore, the first four $\pSat_{\neq 1/2, \{ k \}, \varnothing}$ predicates state that for any edge $(u,v)$,
    the value of $w_k$ for $u$ and $v$ cannot differ, which implies exactly Point~\ref{otori:sign} above.
    Point~\ref{otori:3d} follows directly from the last two $\pSat_{\leq 1/2}$ predicates -- of the two vertices which
    differ only in the $Z$ coordinate, at most one can have $w_k = 0$ and at most one can have $w_k = 1$.

    For the next two points, note that for $j_1, j_2 \neq k$, we have $\text{sat}(w^k, \{1, \ldots, k\} \setminus \{j_1, j_2\}, \varnothing) = 1$
    iff $w^k$ represents one of the labels $\{-j_1, -j_2\}$ -- for positive labels, we have $w_k = 0$ and for negative
    labels other than $-j_1, -j_2$, we have $w_i = 0$ for some $i \notin \{j_1, j_2\}$. Analogously,
    $\text{sat}(w^k, \varnothing, \{1, \ldots, k\} \setminus \{j_1, j_2\}) = 1$ iff $w^k$ represents one of $\{+j_1, +j_2\}$.
    Thus, the four $\pSat_{\leq 1/2, \varnothing, \{ 1, \ldots, k \} \setminus \{j_1, j_2\}}$ predicates within the conjunction over ${j_1, j_2 \in J_1}$
    imply exactly Point~\ref{otori:vertical}, with the conjunction over ${j_1, j_2 \in J_2}$ implying Point~\ref{otori:horizontal}
    analogously.  Clearly, $\pOtori'$ is satisfiable -- an example torus (with coordinate $Z$ omitted) is shown in Figure~\ref{fig:torus}.
\end{proof}

\begin{figure}[t]
    \centering
    \begin{minipage}{.45\textwidth}
    \begin{tikzpicture}[
        every node/.append style={transform shape},
        dot/.style={outer sep=0pt,circle,minimum size=3pt,draw,
        append after command={\pgfextra \draw (\tikzlastnode) ; \endpgfextra}}]
            \matrix (mat) [matrix of nodes,
                         nodes={anchor=center},
                         column sep={0.4em},
                         row sep={0.4em}]{
                             \noWangTile{2.25em}{t00} & \noWangTile{2.25em}{t10} & \noWangTile{2.25em}{t20} & \noWangTile{2.25em}{t30} & \noWangTile{2.25em}{t40} & \noWangTile{2.25em}{t00} \\
                             \noWangTile{2.25em}{t01} & \tileXA{2.25em}{t11} & \tileXB{2.25em}{t21} & \tileXC{2.25em}{t31} & \tileXE{2.25em}{t41} & \noWangTile{2.25em}{t51} \\
                             \noWangTile{2.25em}{t02} & \tileXA{2.25em}{t12} & \tileXD{2.25em}{t22} & \tileXA{2.25em}{t32} & \tileXD{2.25em}{t42} & \noWangTile{2.25em}{t52} \\
                             \noWangTile{2.25em}{t03} & \tileXC{2.25em}{t13} & \tileXE{2.25em}{t23} & \tileXA{2.25em}{t33} & \tileXB{2.25em}{t43} & \noWangTile{2.25em}{t53} \\
                             \noWangTile{2.25em}{t04} & \tileXA{2.25em}{t14} & \tileXD{2.25em}{t24} & \tileXA{2.25em}{t34} & \tileXD{2.25em}{t44} & \noWangTile{2.25em}{t54} \\
                             \noWangTile{2.25em}{t05} & \noWangTile{2.25em}{t15} & \noWangTile{2.25em}{t25} & \noWangTile{2.25em}{t35} & \noWangTile{2.25em}{t45} & \noWangTile{2.25em}{t55} \\
            };
            \foreach \x in {1,...,4}{
                \foreach \y [evaluate=\y as \yy using {int(\y+1)}] in {1,...,4}{
                    \node[draw, fill=t\y\x.rightcolor, circle, minimum size=1em] (x-\y-\x) at ($(t\y\x)!0.5!(t\yy\x)$) {};
                }
            }
            \foreach \x in {1,...,4}{
                \node[draw, fill=t1\x.leftcolor, circle, minimum size=1em] (x-0-\x) at ($(t0\x)!0.5!(t1\x)$) {};
            }

            \foreach \x [evaluate=\x as \xx using {int(\x+1)}] in {0,...,4}{
                \foreach \y [evaluate=\y as \yy using {int(\y+1)}] in {1,...,4}{
                    \node[draw, fill=t\y\x.bottomcolor, circle, minimum size=1em] (y-\y-\x) at ($(t\y\x)!0.5!(t\y\xx)$) {};
                }
            }
            \foreach \y [evaluate=\y as \yy using {int(\y+1)}] in {1,...,4}{
                \node[draw, fill=t\y1.topcolor, circle, minimum size=1em] (y-\y-0) at ($(t\y0)!0.5!(t\y1)$) {};
            }

            \foreach \x [evaluate=\x as \xx using {int(\x+1)}] in {0,...,3}{
                \foreach \y [evaluate=\y as \yy using {int(\y+1)}] in {0,...,3}{
                    \draw (y-\yy-\xx) edge[dashed] (x-\yy-\xx);
                    \draw (x-\yy-\xx) edge[dashed] (y-\yy-\x);
                    \draw (y-\yy-\x) edge[dashed] (x-\y-\xx);
                    \draw (x-\y-\xx) edge[dashed] (y-\yy-\xx);
                }
            }

            \node[fill=black, fill opacity=0.2, rotate fit=45, inner sep=0.5em, fit=(y-1-2) (y-4-2) (x-2-1) (x-2-4)] {};

            \pgfresetboundingbox
            \path[use as bounding box] ($(t11.north west)+(-1.5em,1.5em)$) rectangle ($(t44.south east)+(1.5em,-1.5em)$);
    \end{tikzpicture}
    \end{minipage}
    \begin{minipage}{.45\textwidth}
    \scalebox{0.4}{
    \begin{tikzpicture}[
        every node/.append style={transform shape},
        dot/.style={outer sep=0pt,circle,minimum size=3pt,draw,
        append after command={\pgfextra \draw (\tikzlastnode) ; \endpgfextra}}]
            \matrix (mat) [matrix of nodes,
                         nodes={anchor=center},
                         column sep={0em},
                         ampersand replacement=\&,
                         row sep={0em}]{
                             \hiddenTileXA{3.5em}{tt11} \& \noWangTile{3.5em}{f11} \& \hiddenTileXB{3.5em}{tt21} \& \noWangTile{3.5em}{f21} \& \hiddenTileXC{3.5em}{tt31} \& \noWangTile{3.5em}{f31} \& \hiddenTileXE{3.5em}{tt41} \\
                             \noWangTile{3.5em}{bf11} \& \noWangTile{3.5em}{bf21} \& \noWangTile{3.5em}{bf31} \& \noWangTile{3.5em}{bf41} \& \noWangTile{3.5em}{bf51} \& \noWangTile{3.5em}{bf61} \& \noWangTile{3.5em}{bf71} \\
                             \hiddenTileXA{3.5em}{tt12} \& \noWangTile{3.5em}{f12} \& \hiddenTileXD{3.5em}{tt22} \& \noWangTile{3.5em}{f22} \& \hiddenTileXA{3.5em}{tt32} \& \noWangTile{3.5em}{f32} \& \hiddenTileXD{3.5em}{tt42} \\
                             \noWangTile{3.5em}{bf12} \& \noWangTile{3.5em}{bf22} \& \noWangTile{3.5em}{bf32} \& \noWangTile{3.5em}{bf42} \& \noWangTile{3.5em}{bf52} \& \noWangTile{3.5em}{bf62} \& \noWangTile{3.5em}{bf72} \\
                             \hiddenTileXC{3.5em}{tt13} \& \noWangTile{3.5em}{f13} \& \hiddenTileXE{3.5em}{tt23} \& \noWangTile{3.5em}{f23} \& \hiddenTileXA{3.5em}{tt33} \& \noWangTile{3.5em}{f33} \& \hiddenTileXB{3.5em}{tt43} \\
                             \noWangTile{3.5em}{bf13} \& \noWangTile{3.5em}{bf23} \& \noWangTile{3.5em}{bf33} \& \noWangTile{3.5em}{bf43} \& \noWangTile{3.5em}{bf53} \& \noWangTile{3.5em}{bf63} \& \noWangTile{3.5em}{bf73} \\
                             \hiddenTileXA{3.5em}{tt14} \& \noWangTile{3.5em}{f14} \& \hiddenTileXD{3.5em}{tt24} \& \noWangTile{3.5em}{f24} \& \hiddenTileXA{3.5em}{tt34} \& \noWangTile{3.5em}{f34} \& \hiddenTileXD{3.5em}{tt44} \\
            };
            \foreach \x in {1,...,4}{
                \foreach \y in {1,...,4}{
                    \begin{scope}[shift={(tt\x\y.north west)}]
                        \wangTileClone{tt\x\y}{2em}{tt\x\y-tl}
                        \draw[-{Stealth[scale=3.0]}, color=white, draw opacity=0] (tt\x\y-tl.center) -- (tt\x\y-tl.north west);
                    \end{scope}
                    \begin{scope}[shift={(tt\x\y.south west)}]
                        \wangTileClone{tt\x\y}{2em}{tt\x\y-bl}
                        \draw[-{Stealth[scale=3.0]}, color=white, draw opacity=0] (tt\x\y-bl.center) -- (tt\x\y-bl.south west);
                    \end{scope}
                    \begin{scope}[shift={(tt\x\y.north east)}]
                        \wangTileClone{tt\x\y}{2em}{tt\x\y-tr}
                        \draw[-{Stealth[scale=3.0]}, color=white, draw opacity=0] (tt\x\y-tr.center) -- (tt\x\y-tr.north east);
                    \end{scope}
                    \begin{scope}[shift={(tt\x\y.south east)}]
                        \wangTileClone{tt\x\y}{2em}{tt\x\y-br}
                        \draw[-{Stealth[scale=3.0]}, color=white, draw opacity=0] (tt\x\y-br.center) -- (tt\x\y-br.south east);
                    \end{scope}
                    \draw (tt\x\y-tl) edge[dashed] (tt\x\y-tr);
                    \draw (tt\x\y-tl) edge[dashed] (tt\x\y-bl);
                    \draw (tt\x\y-bl) edge[dashed] (tt\x\y-br);
                    \draw (tt\x\y-tr) edge[dashed] (tt\x\y-br);
                }
            }
            \foreach \x in {1,...,4}{
                \foreach \y [evaluate=\y as \yy using {int(\y+1)}] in {1,...,3}{
                    \draw (tt\x\y-bl) edge[dashed] (tt\x\yy-tl);
                    \draw (tt\x\y-br) edge[dashed] (tt\x\yy-tr);
                }

                \draw[dashed] (tt\x4-bl) to[out=-90, in=-90, out looseness=0.5, in looseness=3.7] ($0.5*(tt\x2-bl)+0.5*(tt\x3-tl)+(0.8,0)$) to[out=90, in=90, in looseness=0.5, out looseness=3.7] (tt\x1-tl);
                \draw[dashed] (tt\x4-br) to[out=-90, in=-90, out looseness=0.5, in looseness=3.7] ($0.5*(tt\x2-br)+0.5*(tt\x3-tr)+(0.8,0)$) to[out=90, in=90, in looseness=0.5, out looseness=3.7] (tt\x1-tr);

            }
            \foreach \y in {1,...,4}{
                \foreach \x [evaluate=\x as \xx using {int(\x+1)}] in {1,...,3}{
                    \draw (tt\x\y-tr) edge[dashed] (tt\xx\y-tl);
                    \draw (tt\x\y-br) edge[dashed] (tt\xx\y-bl);
                }
                \draw[dashed] (tt4\y-br) to[out=0, in=0, out looseness=0.5, in looseness=3.7] ($0.5*(tt2\y-bl)+0.5*(tt3\y-bl)+(0,0.8)$) to[out=180, in=180, in looseness=0.5, out looseness=3.7] (tt1\y-bl);
                \draw[dashed] (tt4\y-tr) to[out=0, in=0, out looseness=0.5, in looseness=3.7] ($0.5*(tt2\y-tl)+0.5*(tt3\y-tl)+(0,0.8)$) to[out=180, in=180, in looseness=0.5, out looseness=3.7] (tt1\y-tl);
            }

            \pgfresetboundingbox
            \path[use as bounding box] ($(tt11-tl.north west)+(-3em,3em)$) rectangle ($(tt44-br.south east)+(3em,-3em)$);
    \end{tikzpicture}
    }
    \end{minipage}
    \caption{
        Left: Li's representation of a $4 \times 4$ torus coloring and the $4 \times 4$ tiling that it yields. The conversion
        from 16 vertices to 32 gives the tiling an additional diagonal periodicity. Combined with the fact that this does
        not work for non-square $n \times m$ tilings (without arranging them into a $\text{lcm}(n, m) \times \text{lcm}(n, m)$ square),
        this proves problematic for restricting the size of a periodic tiling.
        Right: the corresponding $8 \times 8$ torus labeling in our representation. Each tile has 4 labels, one for each corner of the tile (indicated by the arrow in this case). The tiles do not need to be Wang tiles.
    }
    \label{fig:conv1}
\end{figure}
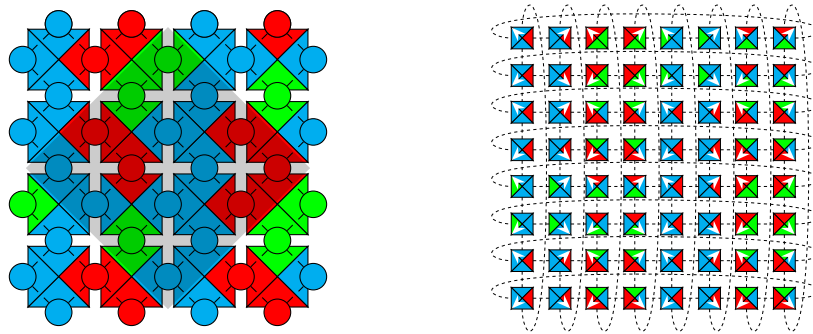
\subsection{Tiles}\label{section:tiles}
Li denotes the set of four possible vertices when $(Z, X_i, Y_j)$ is fixed as a \emph{type $ij$ face} for any $i, j \in \{1, 2\}$ -- e. g. fixing $(Z, X_1, Y_2)$ yields a type 12 face.
The relation between face types and vertex groups is illustrated in Figure~\ref{fig:torus}.

Li's construction utilizes the fact that the four corner-neighbors of any type 11 face are type 22 faces and vice versa.
Because the original construction makes use of a Wang tiling system, these connecting corner vertices can represent the edge colors of the touching
tiles. An example is shown in Figure~\ref{fig:conv1}. The diagonal nature of this tiling proves problematic when we wish to restrict
its size, thus we simply use faces (of type 11) to directly represent tiles, which also allows us to use the general form
of tiling systems. Figure~\ref{fig:conv1} illustrates a torus labeling in our representation.
For given tiling system $\mathcal{D} = (D, H, V)$, we define the following sets:
\begin{align*}
    \mathcal{D}_{11} & = \{(4t+1, 4t+2, 4t+3, 4t) : t \in D\}, \\
    \mathcal{D}_{12} & = \{(4v, 4v+3, 4u+2, 4u+1) : (u, v) \in V\}, \\
    \mathcal{D}_{21} & = \{(4u+2, 4u+1, 4v, 4v+3) : (u, v) \in H\},
\end{align*}
and for each $i \in \{11, 12, 21\}$,
\[
    I_{i} = \{ j_1, \ldots, j_4 \in \{1, \ldots, k-1\} : j_i \ \text{mod} \ 4 \ \text{distinct}, \{j_1, \ldots, j_4\} \notin \mathcal{D}_{i}\}.
\]
The final predicate to enforce that the coloring represents a valid tiling (of size at most $m \times n$) is defined as follows:
\begin{alignat*}{3}
    \pTtori'_{\mathcal{D}} & \ :  \mathrlap{\exists X^2 \leq m, Y^2 \leq n, Z \leq_i 2, W^k \leq_i 2, V^k \leq_i 2, \bar{V}^k \leq_i 2, F \leq_i 2 :} \\
                                & \mathrlap{\pOtori'(X^2, Y^2, Z, W^k, V^k, \bar{V}^k, F)} \\
                            & \land \smash{\bigwedge_{j_1, \ldots, j_4 \in I_{11}}} && \big( \pSat_{\le 3/4, \varnothing, \{1, \ldots, k\} \setminus\{j_1, \ldots, j_4\}}((X_1, Y_1, Z), W^k, V^k, \bar{V}^k, F) \\
                               & && \quad \land \mathrlap{\pSat_{\le 3/4, \{1, \ldots, k\} \setminus \{j_1, \ldots, j_4\}, \varnothing}((X_1, Y_1, Z), W^k, V^k, \bar{V}^k, F)\big)} \\
                               & \land \smash{\bigwedge_{j_1, \ldots, j_4 \in I_{12}}} && \big( \pSat_{\le 3/4, \varnothing, \{1, \ldots, k\} \setminus \{j_1, \ldots, j_4\}}((X_1, Y_2, Z), W^k, V^k, \bar{V}^k, F) \\
                               & && \quad \land \pSat_{\le 3/4, \{1, \ldots, k\} \setminus \{j_1, \ldots, j_4\}, \varnothing}((X_1, Y_2, Z), W^k, V^k, \bar{V}^k, F)\big) \\
                               & \land \smash{\bigwedge_{j_1, \ldots, j_4 \in I_{21}}} && \big( \pSat_{\le 3/4, \varnothing, \{1, \ldots, k\} \setminus \{j_1, \ldots, j_4\}}((X_2, Y_1, Z), W^k, V^k, \bar{V}^k, F) \\
                               & && \quad \land \pSat_{\le 3/4, \{1, \ldots, k\} \setminus \{j_1, \ldots, j_4\}, \varnothing}((X_2, Y_1, Z), W^k, V^k, \bar{V}^k, F)\big).
\end{alignat*}
\begin{lemma}
    $\pTtori'_{\mathcal{D}}$ is satisfied iff $\mathcal{D}$ admits a periodic tiling of size at most $m \times n$.
\end{lemma}
\begin{proof}
    All implicit bounds follow from previously defined predicates used within $\pTtori'_{\mathcal{D}}$: $\pTori'$ for
    $Z$, $\pSw$ for $W^k$, and $\pFlip$ for $V^k$, $\bar{V}^k$, $F$. The bounds of $m, n$ for $X^2, Y^2$ imply that the tori can take any even size up to $2m \times 2n$.

    For the ``only if'' direction, note that the three conjunctions, similarly as in the predicate $\pOtori$, forbid faces of type 11, 12, 21 from being not of
    the form in $\mathcal{D}_{11}, \mathcal{D}_{12}, \mathcal{D}_{21}$, respectively. Clearly, the set $\mathcal{D}_{11}$
    consists of exactly those tiles (represented as type 11 faces) which are in $D$. Similarly, $\mathcal{D}_{12}, \mathcal{D}_{21}$
    consist of all those ``glue'' type 12 and 21 faces which are allowed by $H, V$ respectively. Therefore, in a distribution satisfying $\pTtori'_{\mathcal{D}}$, the type 11 faces represent tiles and neighboring tiles respect the constraints $H$ and $V$. Therefore, each torus
    corresponds to a periodic tiling by $\mathcal{D}$. Since the tori are of size at most $2m \times 2n$, the tiling is of size at most $m \times n$.

    For the ``if'' direction, for any given tiling by $\mathcal{D}$ of size at most $m \times n$, 
    we create a pair of corresponding tori, labeled such that one represents the positive version of this tiling and the
    other the negative version. The satisfiability arguments show that this can be done for any given tiling.
\end{proof}
\subsection{Final construction} \label{section:final}
Once fully expanded, the $\pTtori'_{\mathcal{D}}$ predicate is of the form (omitting the bounds for clarity)
\[
    \exists B, X^2, Y^2, W^k, V^k, \bar{V}^k, F, \ldots : \Big( B \sim \text{Bern}(1/2) \land \bigwedge\nolimits_i (I(A_i ; B_i | C_i) = 0) \Big),
\]
where $A_i, B_i, C_i$ are tuples of the quantified variables. Its negation can be equivalently rewritten:
    \begin{align*}
        & \neg \exists B, \ldots : \big( B \sim \text{Bern}(1/2) \land \bigwedge\nolimits_i (I(A_i ; B_i | C_i) = 0) \big) \\
        \Leftrightarrow \ & \forall B, \ldots : \big( B \not \sim \text{Bern}(1/2) \lor \neg \bigwedge\nolimits_i (I(A_i ; B_i | C_i) = 0) \big) \\
        \Leftrightarrow \ & \forall B, \ldots : \big( \big( \bigwedge\nolimits_i (I(A_i ; B_i | C_i) = 0) \big) \Rightarrow B \not \sim \text{Bern}(1/2) \big) \\
        \Leftrightarrow \ & \forall B, \ldots : \big( \big( \pUnif(B) \land |B| \leq 2 \land \bigwedge\nolimits_i (I(A_i ; B_i | C_i) = 0) \big) \Rightarrow H(B) = 0 \big).
    \end{align*}
The last equivalence holds because a variable $B$ with $|B| \leq 2$ and $\pUnif(B)$ can
either be uniform over one value and have entropy 0 or uniform over two values and have entropy 1.
This final form is a valid CI implication with (partial) cardinality bounds. In the bounded case, the established
cardinality bounds are preserved.

\subsection*{Enforcing the usage of a designated tile}
We extend Li's above construction to enforce that a given tile $t$ be used in the tiling.
Recall from Lemma~\ref{cor:sat} that for any $w^k \in \{0, 1\}^k$,
\[
    H(F|V_{\{i : w_i = 1\}}, \bar{V}_{\{i : w_i = 0\}}, W^k) = \mathbf{P}(W^k = w^k).
\]
Furthermore, variable $F$ is constrained by the predicate $\pUnif_2(F)$. It can be equivalently restated as $\pUnif_{=}(F, B)$, where $\pUnif_{=}$ is defined by Li as follows:
\[
    \pUnif_{=}(Y, Z) : \exists U^3 \leq_i K_Y : \pTriple(Y, U_1, U_2) \land \pTriple(Z, U_1, U_3).
\]
Clearly, $\pUnif_{=}(F, B)$ holds iff $F$ and $B$ are both uniform with the same cardinality and hence $\pUnif_2(F)$ can be replaced by $\pUnif_{=}(F, B)$. Therefore, if the (conditional) entropy of $F$ is nonzero, then the entropy of $B$ must also be nonzero and so we must
have $B \sim \text{Bern}(\frac12)$.

Given a designated tile $t$, let $w^k \in T_k$ be the value of $W^k$ corresponding to label $t$ (without loss of
generality, of vertex group 0 and positive sign) and $S = \{i : w_i = 1\}, \bar{S} = \{i : w_i = 0\}$.
The modified implication is as follows:
\[
    \forall B, \ldots : \big( \big( \pUnif(B) \land |B| \leq 2 \land \bigwedge\nolimits_i (I(A_i ; B_i | C_i) = 0) \big) \Rightarrow H(F|V_S, \bar{V}_{\bar{S}}, W^k) = 0 \big).
\]
The counterexamples of this implication are exactly those labelings which use the tile $t$. Altogether, this chapter
has shown the following theorem:
\begin{theoremSeparate} \label{thm:repr}
    For any given tiling system $\mathcal{D}$ along with tile $t$ and natural numbers $m, n$, there exists a bounded CI implication which
    holds iff $\mathcal{D}$ does not admit a periodic tiling of size at most $m \times n$ which makes use of tile $t$. Moreover, the implication can
    be computed in time polynomial with regard to the size of the tiling system, while the bounds can be computed in time
    polynomial w.r.t. to the size of $m, n$.
\end{theoremSeparate}

Theorem~\ref{thm:repr} gives exactly a polynomial-time many-one reduction from \textsc{Periodic Bounded Tiling} to the
complement of \textsc{Bounded CI Implication}, in particular because $m, n$ and $K$ are encoded in the same manner. In the case of \textsc{Constant-bounded CI Implication}, the above argument does not work since we have constant bounds larger than 2
as well as bounds whose value depends on the input. However, any variable $X$ with cardinality bound $2^j$ can be replaced by the tuple $(X_1, \ldots, X_j)$, where $X_i$
has cardinality bound 2 for each $i \in \{1, \ldots, j\}$. These are clearly equivalent, since each $X_i$ can correspond
to the $i$-th bit of $X$. More generally, a variable with cardinality bound $l$ can be replaced by $(X_1, \ldots, X_{\lceil \log l \rceil})$,
with each $X_i$'s cardinality bounded by 2 and the additional requirement $\pUnif'_k((X_1, \ldots, X_{\lceil \log l \rceil}))$.
Here $\pUnif'_k(Y)$ is a modification of the $\pUnif_k$ predicate such that it enforces $Y$ being uniform with $|Y| \leq k$. The construction of both of these follows closely that of Li and is given in detail in Section~\ref{section:unif}.
Note that the resulting predicate can be large, but this is only important when the bound is not constant -- in our case these are only the two bounds $X^2 \leq m, Y^2 \leq n$.
To avoid this issue, we reduce from \textsc{Power-of-two Periodic Bounded Tiling} -- then $X_i, Y_i$ (for $i \in \{1, 2\}$)
have bounds $2^m, 2^n$ respectively, while the remaining variables have constant bounds.
Replacing $X_1, X_2, Y_1, Y_2$ by tuples of binary variables as shown above, we obtain a CI implication of size
$N \cdot O(m + n)$, where $N$ is the size of the original implication. The number of random variables grows similarly.
The newly created bounds are all constant, and the reduction takes time polynomial with regard to the input size.
The values of the remaining constant bounds are known and therefore each such variable can be converted in constant time.
Together, these two results yield Theorem~\ref{thm:hard}.

\section{Construction of $\pUnif_k$ predicates} \label{section:unif}
We now show how $\pUnif_k$ predicates are constructed, assuming that we have a single random variable $B \sim \text{Bern}(\frac12)$.
Section~\ref{section:final} shows how this assumption is fulfilled once we convert from a predicate to a CI implication.
We annotate Li's original construction by stating the implicit cardinality bounds.

The first defined predicate is the following:
\[
    \pUnif_{=}(Y, Z) : \exists U^3 \leq_i K_Y : \pTriple(Y, U_1, U_2) \land \pTriple(Z, U_1, U_3).
\]

\begin{lemma}
    $\pUnif_{=}(Y, Z)$ is satisfied iff $Y$ and $Z$ are both uniform and have the same cardinality.
\end{lemma}
\begin{proof}
    The ``only if'' direction follows directly from Lemma~\ref{lemma:triple}: $Y$ and $U$ must be uniform with the same
    cardinality as $U_1$. For the ``if'' direction, we let $U_1$ be independent of $(Y, Z)$ and uniform with the same
    cardinality as $Y$ and $Z$. Then we can always construct $U_2$ to satisfy $\pTriple$, for instance by labeling $Y, U_1, U_2$
    with values from $\{1, \ldots, k\}$ where $k$ is the cardinality of $Y$ and setting $U_2 = Y + U_1 \text{ mod } k$.
    The same can be done for $U_3$.
\end{proof}
The following two predicates are defined for multiplication of cardinalities:
\begin{align*}
    \pProd(Y^l, G) : \ & \exists Z_1 \leq_i K_{Y_1}, \ldots, Z_l \leq_i K_{Y_l}, U \leq_i \textstyle{\prod_{i = 1}^l K_{Y_i}} : \\
                            & \pUnif_{=}(G, U) \land H(U | Z^l) = H(Z^l | U) = 0 \\
                            & \land \bigwedge_i ( \pUnif_{=}(Y_i, Z_i) \land I(Z_i; Z^{i-1}) = 0 ), \\
    \pPow_l(Y, G) : \ & \pProd((\underbrace{(Y, \ldots, Y)}_{l}, G).
\end{align*}
\begin{lemma}
    $\pProd((Y_1, \ldots, Y_l), G)$ is satisfied iff all of $Y_1, \ldots, Y_l, G$ are uniform and $|G| = \prod_{i = 1}^l |Y_i|$.
    Consequently, $\pPow_l(Y, G)$ is satisfied iff $Y$ and $G$ are uniform and $|G| = |Y|^l$.
\end{lemma}
\begin{proof}
    For the ``only if'' direction, we have $|Z^l| = \prod_{i = 1}^l |Y_i|$ since $I(Z_i; Z^{i-1}) = 0$
    and $|Z^l| = |U|$ from the two-way functional dependency between $U$ and $Z^l$.
    For the ``if'' direction, we can take $U$ to be a copy of $G$ and then construct $Z^l$ to satisfy the predicate.
\end{proof}
For comparing cardinalities, we have:
\begin{align*}
    \pGesqrt(Y,G) : \ & \exists W \leq_i K_G, U \leq_i K_G, V \leq_i K_Y, Z \leq_i K_Y: \\
                           & \pUnif_{=}(Y, Z) \land \pUnif(W) \land I(W; Z) = 0 \\
                         & \land \pUnif_{=}(G, U) \land H(U | Z, W) = H(Z, W | U) = 0 \\
                         & \land \pUnif_{=}(Z, V) \land H(U | Z, V) = 0, \\
    \pLe(Y, Z) : \ & \exists U \leq_i (K_Y \cdot K_Z) : \pProd((Y, Z), U) \land \pGesqrt(Z, U).
\end{align*}
\begin{lemma}
    $\pGesqrt(Y, G)$ is satisfied iff $Y$ and $G$ are uniform, $|Y|$ divides $|G|$ and $|Y| \geq \sqrt{|G|}$.
\end{lemma}
\begin{proof}
    Denote $|Y| = k, |G| = l$.
    For the ``only if'' direction, we have $|V| = |Z| = k$, thus $|(Z, V)| \leq k^2$.
    Since $|U| = l$ and $H(U|Z, V) = 0$, we obtain $l \leq k^2$.
    From $W$ being uniform and independent of $Z$ and the two-way functional dependency between $U$ and $(Z, W)$, we have
    $l = |W| \cdot k$ and hence $k$ divides $l$.
    For the ``if'' direction, we take $U = (U_1, U_2) \sim \text{Unif}(\{1, \ldots, k\} \times \{1, \ldots, \frac{l}{k}\})$
    and $Z = U_1, V = U_1 + U_2 \text{ mod } l$. Then $Z, V \sim \text{Unif}(\{1, \ldots, k\})$ and $H(U|Z, V) = 0$.
    As the constraints on $W$ can be easily satisfied, the entire predicate is satisfied.
\end{proof}
\begin{lemma}
    $\pLe(Y, Z)$ is satisfied iff $Y$ and $Z$ are uniform and $|Y| \leq |Z|$.
\end{lemma}
\begin{proof}
    The uniformity of $Y$ and $Z$ is clear. We have $\pLe(Y, Z)$ iff $|Z|$ divides ${|Z| \cdot |Y|}$
    and $|Z|^2 \geq |Z| \cdot |Y|$.
\end{proof}
Finally, the predicate $\pUnif_k$ can be constructed. For any $i \in \mathbb{N}, i > 1$, let $p_i, q_i \in \mathbb{N} \setminus \{ 0 \}$ be such that
$\log_2(i - 1) < \frac{p_i}{q_i} < \log_2(i)$ --- such $p_i, q_i$ clearly always exist. Note that since we will only use this for constant $i$, the size of $p_i, q_i$ is not a concern. The predicate $\pUnif_k$ is defined as follows, where $B$ is the given variable distributed as $\text{Bern}(\frac12)$:
\begin{align*}
    \pUnif_k(Y) : \ & \exists U \leq_i 2, V_1 \leq_i 2^{p_k}, V_2 \leq_i 2^{p_{k+1}}, W_1 \leq_i (K_Y)^{q_k}, W_2 \leq_i (K_Y)^{q_{k+1}} : \\
        & \pUnif_{=}(U, B) \land \pUnif(Y) \\
        & \land \pPow_{p_k}(U, V_1) \land \pPow_{q_k}(Y, W_1) \land \pLe(V_1, W_1) \\
        & \land \pPow_{p_{k+1}}(U, V_2) \land \pPow_{q_{k+1}}(Y, W_2) \land \pLe(W_2, V_2).
\end{align*}
\begin{lemma}
    $\pUnif_k(Y)$ is satisfied iff $Y$ is uniform and $|Y| = k$.
\end{lemma}
\begin{proof}
    The uniformity of $Y$ is clear. The definition of $p_i, q_i$ implies that $k$ is the only integer $i$ which
    satisfies $2^{p_k} \leq i^{q_k}$ and $i^{q_{k+1}} \leq 2^{p_{k+1}}$. Given the previous lemmas, this is exactly what
    the predicate $\pUnif_k$ states.
\end{proof}
Note that without $\pLe(V_1, W_1)$, the predicate would be equivalent to $Y$ being uniform with $|Y| \leq k$ -- this is referred to as $\pUnif'_k$ in Section~\ref{section:final}.

\section{Conversion to disjoint CI}\label{section:disjoint}
Thus far, we have considered CI implications without the requirement that the sets of variables within each CI statements
be disjoint.  Li~\cite[Section V-B]{li2021undecidabilityitw} shows a reduction from the general CI implication problem to the disjoint
CI problem, showing that the problem remains undecidable also with this restriction. This is in fact a polynomial many-one reduction. It is shown here and extended to the case with cardinality bounds.
Formally, the following theorem is shown, which implies the second part of Theorem~\ref{thm:hard}.
\begin{theoremSeparate} \label{thm:disjoint}
    Given a vector $R = {\{(A_i, B_i, C_i) : i \in \{1, \ldots, m+1\}\}}$ of triples of subsets of $\{1, \ldots, n\}$
    and vector $K \in (\mathbb{N} \cup \{ \infty \})^n$ of cardinality bounds (with $\infty$ representing unbounded cardinality)
    we can produce in time polynomial relative to the input size a vector $R' = {\{(A'_i, B'_i, C'_i) : i \in \{1, \ldots, m'+1\}\}}$
    of triples of subsets of $\{1, \ldots, n'\}$ and vector $K' \in (\mathbb{N} \cup \{\infty\})^{n'}$ such that for any $i$,
    the sets $A'_i, B'_i, C'_i$ are pairwise disjoint and the implication
    \[
        \bigwedge_{i \in \{1, \ldots, m\}} (I(X_{A_{i}}; X_{B_{i}} | X_{C_{i}}) = 0) \Rightarrow I(X_{A_{m+1}}; X_{B_{m+1}} | X_{C_{m+1}}) = 0
    \]
    holds for all jointly distributed $(X_1, \ldots, X_n)$ with $|X_i| \leq K_i$ if and only if
    \[
        \bigwedge_{i \in \{1, \ldots, m'\}} (I(X_{A'_{i}}; X_{B'_{i}} | X_{C'_{i}}) = 0) \Rightarrow I(X_{A'_{m'+1}}; X_{B'_{m'+1}} | X_{C'_{m'+1}}) = 0
    \]
    holds for all jointly distributed $(X'_1, \ldots, X'_n)$ with $|X'_i| \leq K'_i$. Moreover, the set of values in $K'$
    is the same as that of $K$.
\end{theoremSeparate}

\begin{proof}[Proof of Theorem~\ref{thm:disjoint}]
    We call random variables $Y, Z$ \emph{perfectly resolvable}~\cite{prabhakaran2014assisted} if the predicate $\pRes$, defined
    below along with auxiliary $\pResc$, is satisfied.
\begin{align*}
    \pResc(Y, Z, X) & : I(Y; Z | X) = H(X | Y) = H(X | Z) = 0, \\[3pt]
    \pRes(Y, Z) & : \exists X \leq_i \min(K_Y, K_Z): \pResc(Y, Z, X).
\end{align*}
In fact, $X$ is uniquely defined from $Y, Z$ as it is the Gács-Körner common part of $Y$ and $Z$~\cite{gacs1973common}.
Additionally, $X$ can be assigned a cardinality bound of the minimum of $Y$ and $Z$'s cardinality bounds, since it is a
function of $Y$ as well as on $Z$.
The following predicate is defined, which uses only disjoint CI statements:
\[
    \pResThree(Y^3) : I(Y_1; Y_2 | Y_3) = I(Y_2; Y_3 | Y_1) = I(Y_3; Y_1 | Y_2) = 0.
\]
\begin{lemma}
    $\pRes(Y, Z)$ is satisfied iff the predicate $\exists U \leq_i \min(K_Y, K_Z) : \pResThree(Y, Z, U)$ is satisfied. Further, the latter
    implies that the Gács-Körner common parts of $X$ and $Y$, $X$ and $U$, and $Y$ and $U$ are the same.
\end{lemma}
\begin{proof}
    The ``only if'' is shown by setting $U = X$. The other direction is shown by using the double Markov property~\cite[exercise 16.25]{csiszar2011information}:
    \begin{align*}
        I(U; Y | Z) = I(U; Z | Y) = 0 \Rightarrow \exists & V \leq_i \min(K_Y, K_Z) : \\
        & H(V | Y) = H(V | Z) = I(Y, Z; U | V) = 0.
    \end{align*}
    Combined with $I(Y; Z | U) = 0$, this gives $I(Y; Z | V) = 0$ and thus we can set $X = V$. We also have $H(V | U) = 0$,
    showing that $V$ is the Gács-Körner common part of $X$ and $Y$, $X$ and $U$, and $Y$ and $U$.

    Note that $V$, being a function of $Y$ as well as $Z$, can be bounded to a cardinality of the minimum of $Y$ and $Z$'s
    cardinality bounds.
\end{proof}

In order to express two-way functional dependency using only disjoint CI statements, the following predicate is defined:
\begin{align*}
 \pEqres(Y_1, Z_1, Y_2, Z_2) & : \exists U^2 \leq_i \min(K_{Y_1}, K_{Y_2}) : \pResThree(Y_1, Z_1, U_1) \land \pResThree(Z_1, U_1, U_2) \\
 & \land \pResThree(U_1, U_2, Y_2) \land \pResThree(U_2, Y_2, Z_2).
\end{align*}
\begin{lemma}
    When $\pResc(X_i, Z_i, Y_i)$ holds for both $i \in \{1, 2\}$, then $\pEqres(Y_1, Z_1, Y_2, Z_2)$ is satisfied
    iff $H(X_1 | X_2) = H(X_2 | X_1) = 0$.
\end{lemma}
\begin{proof}
    The ``if'' part is shown by setting $U_1 = U_2 = X_1$ -- then clearly $\pResThree(Y_1, Z_1, U_1)$ and
    $\pResThree(U_2, Y_2, Z_2)$ hold, while the remaining two $\pResThree$ predicates hold because two of the three arguments are
    (functionally) the same. For the ``only if'' part, $\pResThree(Y_1, Z_1, U_1)$ implies
    that the Gács-Körner common part of $Y_1$ and $Z_1$, which is $X_1$, is also the common part of $Z_1$ and $U_1$.
    Note that the $\pResThree$ predicates ``overlap'' -- the last two arguments each one being the first two of the next.
    This allows us to repeat this reasoning, finally yielding that $X_1$ is the common part of $Y_2$ and $Z_2$.
\end{proof}

Finally, given a CI implication of the form
\[
    \bigwedge_{i \in \{1, \ldots, m\}} (I(X_{A_{i}}; X_{B_{i}} | X_{C_{i}}) = 0) \Rightarrow I(X_{A_{m+1}}; X_{B_{m+1}} | X_{C_{m+1}}) = 0
\]
together with cardinality bound $K_i \in (\mathbb{N} \cup \{ \infty \})$ associated with each $X_i$, we construct the
following CI implication over variables $Y^{3n}, Z^{3n}$, with each of the variables $Y_i$, $Y_{i+n}$, $Y_{i+2n}$, $Z_i$, $Z_{i+n}$, $Z_{i+2n}$
assigned the bound $K_i$:
\begin{align*}
 & \bigwedge_{i \in \{1, \ldots, 2n\}} \pEqres(Y_i, Z_i, Y_{i+n}, Z_{i+n}) \\
    \land & \bigwedge_{i \in \{1, \ldots, 3n\}} \big( I(Y_i; Y_{\{1, \ldots, 3n\} \setminus \{ i \}}, Z_{\{1, \ldots, 3n\} \setminus \{ i \}} | Z_i) = 0 \\
          & \qquad \qquad \land I(Z_i; Y_{\{1, \ldots, 3n\} \setminus \{ i \}}, Z_{\{1, \ldots, 3n\} \setminus \{ i \}} | Y_i) = 0 \big) \\
     \land & \bigwedge_{j \in \{1, \ldots, m\}} \big( I(Y_{A_j}; Y_{B_j+n} | Y_{C_j+2n}) = 0 \big) \\
    & \Rightarrow I(Y_{A_{m+1}}; Y_{B_{m+1}+n} | Y_{C_{m+1}+2n}) = 0.
\end{align*}
We now show that the original CI implication (with bounds) holds iff the constructed disjoint CI implication (with bounds) holds.
The ``if'' direction is shown by contraposition --- for any counterexample $X^n$ to the original implication,
setting $Y_i, Y_{i+n}, Y_{i+2n}, Z_i, Z_{i+n}, Z_{i+2n}$ to $X_i$ for each $i \in {\{1, \ldots, n\}}$ gives a counterexample
to the disjoint CI implication.

For the ``only if'' direction, assume that the original implication holds and consider some variables $Y^{3n}, Z^{3n}$
which satisfy the antecedent of the disjoint implication. From the predicate $\pEqres(Y_i, Z_i, Y_{i+n}, Z_{i+n})$ we have that
$Y_i, Z_i$ are perfectly resolvable for all $i \in \{1, \ldots, 3n\}$. Denoting the Gács-Körner common part of $Y_i, Z_i$
by $X_i$, we have
\begin{align}
    & \ I(Y_i; Y_{\{1, \ldots, 3n\} \setminus \{i\}} | X_i) \\
    = & \ I(Y_i; Y_{\{1, \ldots, 3n\} \setminus \{i\}}, Z_{\{1, \ldots, 3n\} \setminus \{i\}} | Z_i, X_i) + I(Y_i; Z_i; Y_{\{1, \ldots, 3n\} \setminus \{i\}}, Z_{\{1, \ldots, 3n\} \setminus \{i\}} | X_i) \\
    \leq & \ I(Y_i; Y_{\{1, \ldots, 3n\} \setminus \{i\}}, Z_{\{1, \ldots, 3n\} \setminus \{i\}} | Z_i, X_i) \label{l0:neg} \\
    = & \ I(Y_i; Y_{\{1, \ldots, 3n\} \setminus \{i\}}, Z_{\{1, \ldots, 3n\} \setminus \{i\}} | Z_i) = 0,
\end{align}
where inequality~\eqref{l0:neg} holds because, denoting $S = \{1, \ldots, 3n\} \setminus \{i\}$,
\[
    I(Y_i; Z_i; Y_S, Z_S | X_i) = \vphantom{\underbrace{a}_{.}} \smash{\underbrace{I(Y_i; Z_i | X_i)}_{= 0}} - \smash{\underbrace{I(Y_i; Z_i | X_i, Y_S, Z_S)}_{\geq 0}}.
\]
Therefore $I(Y_i; Y_{\{1, \ldots, 3n\} \setminus \{i\}} | X_i) = 0$.
It remains to show the following two facts, which essentially say that any $X_i$ and $Y_i$ are equivalent:
\begin{lemma} \label{lemma:xy}
    When $H(X_i | Y_i) = H(X_i | Z_i) = 0$ and $I(Y_i; Y_{\{1, \ldots, 3n\} \setminus \{i\}} | X_i) = 0$ for all $i \in {\{1, \ldots, 3n\}}$,
    then for any $j \in {\{1, \ldots, 3n\}}, A, B, C, D, E, F \subseteq {\{1, \ldots, 3n\}}$ such that $A$ and $B$,
    $C$ and $D$, $E$ and $F$ are disjoint we have
    \[
        I(X_A, Y_B, X_j; X_C, Y_D | X_E, Y_F) = 0 \Leftrightarrow I(X_A, Y_B, Y_j; X_C, Y_D | X_E, Y_F) = 0
    \]
    and
    \[
        I(X_A, Y_B; X_C, Y_D | X_E, Y_F, X_j) = 0 \Leftrightarrow I(X_A, Y_B; X_C, Y_D | X_E, Y_F, Y_j) = 0.
    \]
\end{lemma}

\setcounter{equation}{0}
\begin{proof}[Proof of the first statement]
    The ``if'' direction is immediate from $H(X_j | Y_j) = 0$. For the ``only if'' direction, we have
    \begin{align}
        & I(X_A, Y_B, Y_j; X_C, Y_D | X_E, Y_F) \\
        = \ & I(X_A, Y_B, X_j, Y_j; X_C, Y_D | X_E, Y_F) \label{l1:func} \\
        = \ & I(Y_j; X_C, Y_D | X_A, Y_B, X_E, Y_F, X_j) + \smash{\underbrace{I(X_A, Y_B, X_j; X_C, Y_D | X_E, Y_F)}_{\text{$= 0$ by implication assumption}}} \label{l1:chain}\\
        \leq \ & I(Y_j; X_A, Y_B, X_C, Y_D, X_E, Y_F | X_j) \label{l1:chainin} \\
        \leq \ & I(Y_j; X_A, Y_A, Y_B, X_C, Y_C, Y_D, X_E, Y_E, Y_F | X_j) \label{l1:saturate} \\
        = \ & I(Y_j; Y_A, Y_B, Y_C, Y_D, Y_E, Y_F | X_j) \label{l1:pullback} \\
        \leq \ & I(Y_j; Y_{\{1, \ldots, 3n\} \setminus \{j\}} | X_j) \; = \; 0. \label{l1:saturatemore}
    \end{align}

    Equality~\eqref{l1:func} holds since we add $X_j$ which depends functionally on $Y_j$.
    Equality~\eqref{l1:chain} is an application of the chain rule for mutual information~\cite{yeung2002first}
    $I(U, V; W) = I(U; W | V) + I(V; W)$ with $U = Y_j, V = (X_A, Y_B, X_j), W = (X_C, Y_D)$ and also conditioned everywhere
    on $(X_E, Y_F)$. Inequality~\eqref{l1:chainin} follows again from the chain rule. For inequality~\eqref{l1:saturate},
    we add corresponding $Y_i$ for each $X_i$, and clearly ${I(U, V; W) = I(U; W | V) + I(V; W) \geq I(V; W)}$.
    Equality~\eqref{l1:pullback} holds similarly to~\eqref{l1:func} due to functional dependencies and equality~\eqref{l1:saturatemore}
    similarly to~\eqref{l1:saturate} simply adds more variables. The final equality holds by assumption.
\end{proof}

\setcounter{equation}{0}
\begin{proof}[Proof of the second statement]
    The ``if'' direction is immediate from $H(X_j | Y_j) = 0$. For the ``only if'' direction, we have
    \begin{align}
        & I(X_A, Y_B; X_C, Y_D | X_E, Y_F, Y_j) \\
        = \ & I(X_A, Y_B; X_C, Y_D | X_E, Y_F, X_j, Y_j) \label{l2:func} \\
        \leq \ & I(X_A, Y_B, Y_j; X_C, Y_D | X_E, Y_F, X_j) \label{l2:chainin} \\
        = \ & I(Y_j; X_C, Y_D | X_A, Y_B, X_E, Y_F, X_j) + \smash{\underbrace{I(X_A, Y_B; X_C, Y_D | X_E, Y_F, X_j)}_{\text{$= 0$ by implication assumption}}} \label{l2:chain} \\
        \leq \ & I(Y_j; X_A, Y_B, X_C, Y_D, X_E, Y_F | X_j) \\
        \leq \ & I(Y_j; X_A, Y_A, Y_B, X_C, Y_C, Y_D, X_E, Y_E, Y_F | X_j) \\
        = \ & I(Y_j; Y_A, Y_B, Y_C, Y_D, Y_E, Y_F | X_j) \\
        \leq \ & I(Y_j; Y_{\{1, \ldots, 3n\} \setminus \{j\}} | X_j) = 0.
    \end{align}
    The reasoning is very similar to the proof above. Equality~\eqref{l2:func} holds due to $X_j$ being functionally
    dependent on $Y_j$, while inequality~\eqref{l2:chainin} and equality~\eqref{l2:chain} follow from the chain rule for mutual information.
    The following steps are exactly the same as in the proof above.
\end{proof}

Using Lemma~\ref{lemma:xy}, from each $I(Y_{A_j}; Y_{B_j+n} | Y_{C_j+2n}) = 0$ we get
$I(X_{A_j}; X_{B_j+n} | X_{C_j+2n}) = 0$.
From the two-way functional dependencies between $X_i, X_{i+n}, X_{i+2n}$ and the original implication we have
$I(X_{A_{m+1}}; X_{B_{m+1}} | X_{C_{m+1}}) = 0$. Using the functional dependencies and the above lemma again yields
$I(Y_{A_{m+1}}; Y_{B_{m+1}+n} | Y_{C_{m+1}+2n}) = 0$.
\end{proof}

\bibliographystyle{plainurl}
\bibliography{ci-conexptime}{}

\begin{thebibliography}{10}

\bibitem{canny1988pspace}
John~F. Canny.
\newblock Some algebraic and geometric computations in {PSPACE}.
\newblock In Janos Simon, editor, {\em Proceedings of the 20th Annual {ACM}
  Symposium on Theory of Computing, May 2-4, 1988, Chicago, Illinois, {USA}},
  pages 460--467. {ACM}, 1988.
\newblock \href {https://doi.org/10.1145/62212.62257}
  {\path{doi:10.1145/62212.62257}}.

\bibitem{csiszar2011information}
Imre Csiszár and János Körner.
\newblock {\em Information Theory: Coding Theorems for Discrete Memoryless
  Systems}.
\newblock Cambridge University Press, 2 edition, 2011.
\newblock \href {https://doi.org/10.1017/CBO9780511921889}
  {\path{doi:10.1017/CBO9780511921889}}.

\bibitem{geiger1993logical}
Dan Geiger and Judea Pearl.
\newblock Logical and algorithmic properties of conditional independence and
  graphical models.
\newblock {\em The Annals of Statistics}, 21(4):2001--2021, 1993.

\bibitem{gottesman2009quantum}
Daniel Gottesman and Sandy Irani.
\newblock The quantum and classical complexity of translationally invariant
  tiling and hamiltonian problems.
\newblock {\em Proceedings - Annual IEEE Symposium on Foundations of Computer
  Science, FOCS}, 05 2009.
\newblock \href {https://doi.org/10.1109/FOCS.2009.22}
  {\path{doi:10.1109/FOCS.2009.22}}.

\bibitem{gacs1973common}
Péter Gács and J.~Körner.
\newblock Common information is far less than mutual information.
\newblock {\em Problems of Control and Information Theory}, 2, 01 1973.

\bibitem{hannula2019facets}
Miika Hannula, {\AA}sa Hirvonen, Juha Kontinen, Vadim Kulikov, and Jonni
  Virtema.
\newblock Facets of distribution identities in probabilistic team semantics.
\newblock In Francesco Calimeri, Nicola Leone, and Marco Manna, editors, {\em
  Logics in Artificial Intelligence - 16th European Conference, {JELIA} 2019,
  Rende, Italy, May 7-11, 2019, Proceedings}, volume 11468 of {\em Lecture
  Notes in Computer Science}, pages 304--320. Springer, 2019.
\newblock \href {https://doi.org/10.1007/978-3-030-19570-0\_20}
  {\path{doi:10.1007/978-3-030-19570-0\_20}}.

\bibitem{khamis2020decision}
Mahmoud~Abo Khamis, Phokion~G. Kolaitis, Hung~Q. Ngo, and Dan Suciu.
\newblock Decision problems in information theory.
\newblock In Artur Czumaj, Anuj Dawar, and Emanuela Merelli, editors, {\em 47th
  International Colloquium on Automata, Languages, and Programming, {ICALP}
  2020, July 8-11, 2020, Saarbr{\"{u}}cken, Germany (Virtual Conference)},
  volume 168 of {\em LIPIcs}, pages 106:1--106:20. Schloss Dagstuhl -
  Leibniz-Zentrum f{\"{u}}r Informatik, 2020.
\newblock \href {https://doi.org/10.4230/LIPIcs.ICALP.2020.106}
  {\path{doi:10.4230/LIPIcs.ICALP.2020.106}}.

\bibitem{kuhne2022entropic}
Lukas K{\"{u}}hne and Geva Yashfe.
\newblock On entropic and almost multilinear representability of matroids.
\newblock {\em CoRR}, abs/2206.03465, 2022.
\newblock \href {https://arxiv.org/abs/2206.03465} {\path{arXiv:2206.03465}}.

\bibitem{lewis1998elements}
Harry~R. Lewis and Christos~H. Papadimitriou.
\newblock {\em Elements of the theory of computation, 2nd Edition}.
\newblock Prentice Hall, 1998.

\bibitem{li2021undecidabilityitw}
Cheuk~Ting Li.
\newblock The undecidability of conditional affine information inequalities and
  conditional independence implication with a binary constraint.
\newblock {\em {IEEE} Transactions on Information Theory}, 68(12):7685--7701,
  2022.
\newblock \href {https://doi.org/10.1109/TIT.2022.3190800}
  {\path{doi:10.1109/TIT.2022.3190800}}.

\bibitem{li2022undecidability}
Cheuk~Ting Li.
\newblock Undecidability of network coding, conditional information
  inequalities, and conditional independence implication.
\newblock {\em {IEEE} Transactions on Information Theory}, 2023.
\newblock \href {https://doi.org/10.1109/TIT.2023.3247570}
  {\path{doi:10.1109/TIT.2023.3247570}}.

\bibitem{prabhakaran2014assisted}
Vinod~M. Prabhakaran and Manoj Prabhakaran.
\newblock Assisted common information with an application to secure two-party
  sampling.
\newblock {\em {IEEE} Trans. Inf. Theory}, 60(6):3413--3434, 2014.
\newblock \href {https://doi.org/10.1109/TIT.2014.2316011}
  {\path{doi:10.1109/TIT.2014.2316011}}.

\bibitem{rothemund2000squares}
Paul W.~K. Rothemund and Erik Winfree.
\newblock The program-size complexity of self-assembled squares (extended
  abstract).
\newblock In {\em Proceedings of the Thirty-Second Annual ACM Symposium on
  Theory of Computing}, STOC '00, page 459–468, New York, NY, USA, 2000.
  Association for Computing Machinery.
\newblock \href {https://doi.org/10.1145/335305.335358}
  {\path{doi:10.1145/335305.335358}}.

\bibitem{boas1996tilings}
Peter van Emde~Boas.
\newblock The convenience of tilings.
\newblock In A.~Sorbi, editor, {\em Complexity, Logic, and recursion Theory},
  volume 187 of {\em Lecture Notes in Pure and Applied Logic}, pages 331--363.
  Marcel Dekker, Inc., 1997.

\bibitem{yeung2002first}
R.W. Yeung.
\newblock {\em A First Course in Information Theory}.
\newblock Information Technology: Transm. Springer US, 2002.

\bibitem{zhang1997ineq}
Z.~Zhang and R.W. Yeung.
\newblock A non-{Shannon}-type conditional inequality of information
  quantities.
\newblock {\em IEEE Transactions on Information Theory}, 43(6):1982--1986,
  1997.
\newblock \href {https://doi.org/10.1109/18.641561}
  {\path{doi:10.1109/18.641561}}.

\end{thebibliography}

\end{document}